\documentclass[11pt,reqno]{article}
\usepackage[margin=1in]{geometry}
\usepackage{dsfont, amssymb,amsmath,amscd,latexsym, amsthm, amsxtra,amsfonts}
\usepackage{lineno}
\usepackage[all]{xy}
\usepackage[active]{srcltx}

\usepackage{tikz}
\usepackage{lipsum}
\usepackage{amsfonts}
\usepackage{graphicx}
\usepackage{epstopdf}
\usepackage{algorithm}
\usepackage{algorithmic}
\usepackage{xcolor}
\usepackage{booktabs}
\usepackage{array}
\usepackage[round]{natbib}
\usepackage{bbm}
\usepackage{enumerate}
\usepackage{mathrsfs}
\usepackage{graphicx}
\usepackage{comment}
\usepackage{mathtools}
\usepackage{cases}
 \usepackage{tikz}
\usetikzlibrary{calc,arrows}
\usepackage{verbatim}
\usepackage{graphicx}
\usepackage{color}

\usepackage{tikz}
\usetikzlibrary{calc,arrows}
\usepackage{verbatim}
\usepackage{graphicx}
\usepackage{subfigure}
\usepackage{epstopdf}
\usepackage{verbatim}
\usepackage{graphicx}
\usepackage{color}
\usepackage[affil-it]{authblk}

\newtheorem{theorem}{Theorem}[section]
\newtheorem{corollary}[theorem]{Corollary}
\newtheorem{proposition}[theorem]{Proposition}
\newtheorem{lemma}[theorem]{Lemma}

\theoremstyle{definition}
\newtheorem{definition}[theorem]{Definition}

\theoremstyle{definition}

\theoremstyle{definition}

\theoremstyle{definition}
\newtheorem{assumption}{Assumption}

\renewcommand{\epsilon}{\varepsilon}

\usepackage[pdfstartview=FitH, bookmarksnumbered=true,bookmarksopen=true, colorlinks=true, pdfborder=001, citecolor=blue, linkcolor=blue,urlcolor=blue]{hyperref}
\usepackage{graphics}
\graphicspath{{figures/}}
\usepackage{amsopn}

\DeclareMathOperator{\var}{Var}
 \title{Stackelberg reinsurance and premium decisions with MV criterion and irreversibility}
 
  \author{Zongxia Liang\thanks{Email: \texttt{liangzongxia@mail.tsinghua.edu.cn}}\ \ \ \ \  Xiaodong Luo\thanks{Email: \texttt{luoxd21@mails.tsinghua.edu.cn}}
 }	
\affil{Department of Mathematical Sciences, Tsinghua University, China}
\date{}

\begin{document}
\maketitle
\begin{abstract}
We study a reinsurance Stackelberg game in which both the insurer and the reinsurer adopt the mean-variance (abbr. MV) criterion in their decision-making and the reinsurance is irreversible. We apply a unified singular control framework where irreversible reinsurance contracts can be signed in both discrete and continuous times. The results theoretically illustrate that, rather than continuous-time contracts or a bunch of discrete-time contracts, a single once-for-all reinsurance contract is preferred. Moreover, the Stackelberg game turns out to be centering on the signing time of the single contract. The insurer signs the contract if the premium rate is lower than a time-dependent threshold and the reinsurer designs a premium that triggers the signing of the contract at his preferred time. Further, we find that reinsurance preference, discount and reversion have a decreasing dominance in the reinsurer's decision-making, which is not seen for the insurer.
\vskip 10 pt \noindent
2020
{\bf Mathematics Subject Classification:}  91A65, 91G05, 93E20
 \vskip 10pt  \noindent
{\bf Keywords:} Stackelberg game; singular control; MV criterion; time inconsistency; reinsurance and premium strategies.
\vskip 5pt\noindent
\end{abstract}

\section{Introduction}

Reinsurance is an effective tool for insurers to transfer their risk to reinsurers. In practice, the insurer pays a reinsurance premium for passing part of the risk to the reinsurer. Determining the optimal amount of transferred risk or developing an optimal reinsurance strategy is an essential problem for insurers to have better business capacity and solvency.

The modeling of reinsurance or reinsurance contracts is an essential topic in financial mathematics and actuarial science. It has long been a dilemma that the risk or surplus should be measured in continuous time while the reinsurance contracts are signed at some discrete time points. On the one hand, the major stream of research on reinsurance uses continuous-time models where the reinsurance contracts are signed or adjusted in continuous times. However, the theoretical results can hardly be applied due to the infeasibility of continuous actions on reinsurance contracts. On the other hand, using discrete-time models is at the cost of approximating the risk or surplus in discrete times; see \cite{bazyari2023ruin} and \cite{boonta2019approximation}. There is also literature that combines continuous-time risk or surplus and discrete-time reinsurance contracts; see \cite{dickson2006optimal} and \cite{federico2023irreversible}. However, none of the literature theoretically justifies the necessity of discrete-time reinsurance contracts under a continuous-time formulation of risk or surplus.

In this paper, we propose a theoretical Stackelberg game model with MV criterion in a unified singular control framework considering both discrete-time and continuous-time reinsurance contracts, and the reinsurance is irreversible. In the Stackelberg game, the reinsurer is the leader who offers many reinsurance contracts presented by the reinsurance premium rates, while the insurer is the follower who can sign reinsurance contracts at both continuous and discrete times. The objectives of both the insurer and the reinsurer are of MV type and the reinsurance contracts cannot be reversed. The insurer's problem is a time-inconsistent singular control problem while the reinsurer's problem inherits the singularity of the follower's problem through the Stackelberg game; see Subsection \ref{subsec-es}. For the insurer's problem, we obtain the extended HJB equations through a technical variational argument. Then, we make an ansatz of the solution structure, from which we explicitly obtain a solution. The solution is verified by a dedicated verification theorem. Then, analyzing the equivalent equilibrium conditions, we solve the reinsurer's problem. The reinsurer's problem can be transformed into an explicitly solvable time-inconsistent time-selection problem, whose equilibrium is similar to those of time-inconsistent stopping problems; see \cite{bjork2021time} and references therein.

The equilibrium solutions of our reinsurance Stackelberg game have significant financial and actuarial implications.

First, the equilibrium reinsurance strategy is to purchase a lump sum of all remaining reinsurance reserves the first time when the price becomes lower than a time-dependent threshold. It implies that, under MV criterion, signing reinsurance contracts in discrete times is preferred than signing reinsurance contracts continuously. Moreover, signing a single once-for-all reinsurance contract at a carefully chosen time is preferred than signing a bunch of reinsurance contracts at different times. 

Second, the equilibrium premium strategy is to set the premium exactly the threshold that triggers a reinsurance contract if the reinsurer has an intention to sell reinsurance, while the premium can be arbitrarily set above the threshold if he has no intention. It implies that the preferred reinsurance premium should be designed to trigger the reinsurance contract at the time desired by the reinsurer. Moreover, the time to sign the single once-for-all contract is at the core of decision-making for both the insurer and the reinsurer.

Third, we find that reinsurance preferences, reversions and discounts of the insurer and the reinsurer are key factors that influence the equilibrium reinsurance and premium strategies. Fixing any reinsurance premium, the equilibrium reinsurance strategy of the insurer is only affected by her own reinsurance preference, reversion and discount. Enlarging any of these three factors postpones the time to sign the single once-for-all reinsurance contract. The reinsurer's equilibrium premium strategy is affected by all six factors, i.e., reinsurance preferences, reversions and discounts of the insurer and the reinsurer. Further, comparing the influences of these parameters on the reinsurer's intention to sell reinsurance, we find that reinsurance preference is more dominant than discount while discount is more dominant than reversion.

The major contributions of this paper are as follows.

First, we study a reinsurance Stackelberg game in the singular control approach. The singular control approach provides a unified framework for both discrete-time and continuous-time reinsurance under irreversibility. Existing literature on reinsurance Stackelberg games is largely restricted to continuous-time reinsurance where there is no irreversibility. However, despite the feasibility of reversing reinsurance contracts, continuously signing or adjusting the reinsurance contract is impractical. Our finding theoretically suggests that, under the MV criterion with irreversibility, a discrete-time once-for-all reinsurance is preferred and the reinsurance Stackelberg game between the insurer and the reinsurer centers on the best time to sign the contract.

Mathematically, studying the reinsurance Stackelberg game in a singular control framework is difficult. Explicit solutions to the extended HJB equations of singular control problems are difficult to obtain.   Let alone solve the leader's problem in addition to the follower's singular control problem, where the optimal/equilibrium singular control should be substituted into the leader's problem. A major obstacle lies in handling the smoothness of solutions on the boundary of an immediate jump. However, this issue is tackled by the special structure of our ansatz with properly weakened smoothness conditions and we explicitly solve both problems of the follower and the leader.

Second, in the insurer's problem, we provide an instance of a finite-time singular control solution where the jump, i.e., immediate reinsurance purchase or immediate execution of an action, can occur after the initial time. The phenomenon is economically plausible in common sense. However, such cases are rare in singular control literature for an optimal or equilibrium solution. The reason lies in that the waiting regions and the purchasing regions in their models are time-independent; see \cite{yan2022irreversible} and \cite{ferrari2021optimal} for example. In such cases, the optimal or equilibrium solution that solves the corresponding Skorohod reflection problem can only jump at the initial time. The two direct causes of time-independence are: first, their models are infinite-time models; second, the pre-discount cost rate in the penalty term of singular control is time-independent, i.e., time-independent reinsurance premium rate in irreversible reinsurance. Both causes do not suit our model because the MV criterion requires the finite-time horizon and the Stackelberg game requires the reinsurance premium rate as the reinsurer's strategy to be time-dependent. The time to sign the reinsurance contract turns out to be the most important focus of our reinsurance Stackelberg game.

Technically, finite-time horizon and time-dependence cause great difficulty in obtaining an explicit solution even for regular control problems. The reason is that finite-time horizon and time-dependence result in HJB equations of PDE types while the equations are ODEs otherwise. We resolve these issues even in a singular control problem by properly making an ansatz and using a dedicated verification argument.

Third, the solution to the insurer's problem is the first explicit solution to a time-inconsistent singular control problem under the equilibrium framework of \cite{liang2023equilibria}. The solution provides an example of the existence of the equilibrium for a time-inconsistent singular control problem. \cite{liang2023equilibria} theoretically proves the existence of equilibrium for the problems of the non-exponential discount type under additional conditions. The time-inconsistency of the insurer's problem in this paper is caused by the MV criterion, which has a different cause but again supports the existence of such kind of equilibrium. 

The rest of the paper is organized as follows: In Section \ref{sec-model}, we introduce our model of reinsurance Stackelberg game with MV criterion and irreversibility and define the equilibrium strategies of the insurer and the reinsurer. We solve the insurer's problem in Section \ref{sec-insurer} and obtain the insurer's equilibrium reinsurance strategy. In Section \ref{sec-reinsurer}, we solve the reinsurer's problem and obtain the reinsurer's equilibrium premium strategy. We give an interpretation of the equilibrium strategies of the insurer and the reinsurer in Section \ref{sec-interpretation}. We conclude in Section \ref{sec-conclusion}.

\subsection{Literature review}
In recent years, the research on reinsurance has largely incorporated the theory of stochastic games, especially Stackelberg game. In addition to finding the optimal reinsurance strategy for the insurer, the Stackelberg game approach also gives ideas for the reinsurer on how to price the reinsurance premium. Dating back to \cite{tapiero1981optimum}, there is a vast literature on the reinsurance Stackelberg game; see \cite{bai2021stackelberg},  \cite{yang2022robust} and \cite{gavagan2022optimal}. The reinsurance Stackelberg game has been extended recently by considering multiple insurers or reinsurers; see \cite{bai2022hybrid}, \cite{wang2023optimal}, \cite{zhu2023equilibria} and \cite{cao2023stackelberg}. There are other types of stochastic games applied in reinsurance problems, including cooperative and competitive games; see \cite{yang2024time}, \cite{yang2021equilibrium} and \cite{jiang2019optimal}.

A major extension to the reinsurance Stackelberg game focuses on the equilibrium solutions under time-inconsistency. The equilibrium here refers to the time-inconsistency equilibrium in addition to the Stackelberg game equilibrium. In practice, the time-inconsistency is usually caused by the MV criterion, where the decision-maker usually maximizes the mean while minimizes the variance. For an optimization problem with an MV-type objective, the optimal solution now is not optimal the next moment and one can only seek an equilibrium solution rather than an optimal solution. The literature of this stream includes \cite{li2022stackelberg}, \cite{chen2019stochastic} and \cite{han2023optimal}. Recent studies finds that the Stackelberg game structure can also lead to time-inconsistency, which strengthens the importance of equilibrium solutions in Stackelberg games; see \cite{zhou2023time} and references therein.

Recently, there has been an emerging stream of research focusing on the irreversibility of reinsurance. The irreversibility of reinsurance requires that the reinsurance contracts can not be reversed. The pioneering work \cite{yan2022irreversible} models the irreversibility by considering the transferred risk as a singular control. Another piece of research \cite{federico2023irreversible} models the irreversibility by adding a starting time of the reinsurance contract as a control variable.

In the singular control approach, the controls are formed in continuous times to incorporate continuous actions while the discrete actions are presented by the jumps of the controls. The main stream of literature on reinsurance applies the continuous-time formulation, which only accounts for continuous-time reinsurance. There are also discrete-time formulations (\cite{bazyari2023ruin},  \cite{boonta2019approximation}) and continuous-time formulations with discrete-time reinsurance (\cite{dickson2006optimal},  \cite{federico2023irreversible}). However, these approaches cannot put both continuous-time and discrete-time reinsurance contracts together for comparison. \cite{yan2022irreversible} applies the singular control approach, but their result suggests a preference for hybrid reinsurance, which is different from our intuitively more reasonable discrete-time reinsurance. We give a detailed comparison in the next paragraph.

Let us mention two pieces of literature most related to this paper. One is \cite{yan2022irreversible} which first applies a singular control model to study reinsurance with irreversibility. The formulation of risk exposure and accumulated transferred risk as a singular control in our model is inspired by the model in \cite{yan2022irreversible}. We also adopt a reinsurance ceiling as in \cite{yan2022irreversible} but such constraint on reinsurance is not rare; see \cite{jiang2019optimal} and the references in \cite{yan2022irreversible}. However, we are different in the following aspects which makes our problem more challenging and leads to intuitively more reasonable results. First, our problem is planned in finite time which is more reasonable because the duration of (irreversible) reinsurance contracts is very long and is at least longer than the planning time. Second, we further consider the reinsurer's problem through the Stackelberg game, which forces the reinsurance premium rate to depend on time. Third, the well-known MV criterion leads to time-inconsistency which is new to combine with singular control. 

The other piece of closely related literature is \cite{liang2023equilibria} which theoretically establishes an equilibrium framework for time-inconsistent singular control problem. We apply the equilibrium framework in \cite{liang2023equilibria} to study the insurer's problem. The derivation of extended HJB equations and existence results in \cite{liang2023equilibria} is limited to non-exponential discount, which does not directly apply to time-inconsistency of MV type. Moreover, we obtain an explicit equilibrium solution, which is presumably difficult for time-inconsistent singular control problems. Our explicit solution supports the general existence of the equilibrium in \cite{liang2023equilibria}.

\section{Model Formulation}
\label{sec-model}
Let $(\Omega,\mathcal{F},\{\mathcal{F}_{t}\}_{t\ge 0}, \mathbb{P})$ be a filtered probability space satisfying the usual conditions. There are two independent standard one-dimensional Brownian motions $\{B^{I}_{t}\}_{t\ge 0}$ and $ \{B^{R}_{t}\}_{t\ge 0}$ on this space.

Denote the uncontrolled risk exposure processes of the insurer and the reinsurer by $\{X_{t}\}_{t\ge 0}$ and $\{Z_{t}\}_{t\ge 0}$ respectively. We assume that both risk exposure processes follow the mean-reverting dynamic,
\begin{align*}
&\left\{
\begin{array}{l}
dX_{t}=(a^{I}-b^{I}X_{t})dt+\sigma^{I} dB^{I}_{t},\ t\ge 0,\\
X_{0-}=x_{0},
\end{array}
\right.\\
&\left\{
\begin{array}{l}
dZ_{t}=(a^{R}-b^{R}Z_{t})dt+\sigma^{R} dB^{R}_{t},\ t\ge 0,\\
Z_{0-}=z_{0},
\end{array}
\right.
\end{align*}
where $a^{I}, a^{R}, b^{I}, b^{R}, \sigma^{I}, \sigma^{R}$ are all endogenously given constants. The reversion coefficients $b^{I}, b^{R}$ and the volatility coefficients $\sigma^{I}, \sigma^{R}$ are all positive.

The insurer can irreversibly transfer its risk exposure to the reinsurer by paying a reinsurance premium to the reinsurer. The reinsurance premium rate is issued by the reinsurer. Denote the accumulated amount of risk exposure covered by reinsurance up to time $t$ by $\xi=\{\xi_{t}\}_{t\ge 0}$, which is non-decreasing and c$\acute{a}$dl$\acute{a}$g with initial $\xi_{0-}=y_{0}\ge 0$. We assume that the reinsurance coverage has a ceiling (maximum amount) $\bar{y}>0$, which is an endogenously given constant. Such ceiling on reinsurance is natural and not rare in literature; see \cite{jiang2019optimal}, \cite{yan2022irreversible} and references therein.

The controlled risk exposure processes of the insurer and the reinsurer, denoted by $X^{\xi}$ and $Y^{\xi}$ respectively, follow the stochastic differential equations
\begin{align*}
&\left\{
\begin{array}{l}
dX^{\xi}_{t}=(a^{I}-b^{I}X^{\xi}_{t})dt+\sigma^{I} dB^{I}_{t}-d\xi_{t},\ t\ge 0,\\
X^{\xi}_{0-}=x_{0},
\end{array}
\right.\\
&\left\{
\begin{array}{l}
dZ^{\xi}_{t}=(a^{R}-b^{R}Z^{\xi}_{t})dt+\sigma^{R} dB^{R}_{t}+d\xi_{t},\ t\ge 0,\\
Z^{\xi}_{0-}=z_{0}.
\end{array}
\right.
\end{align*}

\subsection{The Stackelberg game and admissible strategies}

In this subsection, we state the Stackelberg game between the insurer and the reinsurer and introduce their admissible strategies respectively.

We formulate the reinsurance problem as a Stackelberg differential game. The problem is considered on a finite time horizon $[0, T]$ with constant $T>0$. The reinsurer is the leader that offers a reinsurance premium rate $c=\{c(t)\}_{t\in[0, T]}$ from the set of all admissible reinsurance premium rates $\mathcal{C}$, which is the set of all nonnegative processes. We refer to $c$ as the reinsurer's premium strategy. The insurance company is the follower that chooses a reinsurance strategy $\xi\in\mathcal{D}_{[0, T]}$ based on the offered reinsurance premium rate, where $\mathcal{D}_{[0, T]}$ is the set of all admissible reinsurance strategies on $[0, T]$. We define the admissible reinsurance strategies in Definition \ref{ads}.

Given the reinsurer's premium strategy $c$, the insurer chooses a reinsurance strategy $\xi(c)$ to minimize
\begin{equation*}
J^{I}(x,t,y;\xi):=\mathbb{E}_{x,t,y}X^{\xi}_{T}+\gamma^{I}\var_{x,t,y}X^{\xi}_{T}+\theta^{I}\mathbb{E}_{x,t,y}\int_{t}^{T}e^{\rho^{I}(T-r)}c(r)d\xi_{r}.
\end{equation*}
The positive constants $\gamma^{I} $ and $\theta^{I}$ balance the three objective terms, which reflects the insurer's preferences on variance and reinsurance. The constant $\rho^{I}$ is the insurer's time discount factor. The integration w.r.t $d\xi$ in the last term is defined by
\begin{equation*}
\int_{t}^{T}c(r)d\xi_{r}:=\int_{t}^{T}c(r)d\xi^{c}_{r}+\sum\limits_{r\in[t,T]}c(r)\Delta\xi_{r},
\end{equation*}
where $\xi^{c}_{r}$ and $\Delta\xi_{r}:=\xi_{r}-\xi_{r-}$ are the continuous and discrete parts of $\xi_{r}$.

Denote the state-time-control space of the insurer's problem by 
\begin{equation*}
\mathcal{Q}:=\mathbb{R}\times[0,T]\times[0,\bar{y}].
\end{equation*}
The definition of an admissible reinsurance strategy is as follows.
\begin{definition}[admissible reinsurance strategy]
	\label{ads}
	Given initial $(x,t,y)\in\mathcal{Q}$, suppose that the process $\{\xi_{r}\}_{r\in[t,T]}$ is an $\mathcal{F}$-adapted, non-decreasing,  c$\acute{a}$dl$\acute{a}$g  process and satisfies\\
	(a) $\xi_{t-}=y$, $\xi_{T}\le\bar{y}$ and the stochastic differential equation
	\begin{equation*}
	\left\{
	\begin{array}{l}
dX^{\xi}_{r}=(a^{I}-b^{I}X^{\xi}_{r})dr+\sigma^{I}dB^{I}_{r}-d\xi_{r},\ r\in[t,T],\\
	X^{\xi}_{t-}=x.
	\end{array}
	\right.
	\end{equation*}
	has a unique strong  solution $\{X^{\xi}_{r}\}_{r\in[t,T]}$.\\
	(b) The $\{\xi_{r}\}_{r\in[t,T]}$ and $\{X^{\xi}_{r}\}_{r\in[t,T]}$ given in (a) satisfy
	\begin{align*}
	&\mathbb{E}_{x,t,y}(X^{\xi}_{T})^{2}<\infty,\\
	&\mathbb{E}_{x,t,y}\int_{t}^{T}e^{\rho^{I}(T-r)}c(r)d\xi_{r}<\infty.
	\end{align*}
(c) For any $t\in[0,T)$, $\Delta\xi_{t}$ is $\mathcal{F}_{t-}$ measurable and $\sum\limits_{r\in(t,t+h)}\Delta\xi_{r}=o(h)$ as $h\rightarrow 0$.
 
	Then we call $\{\xi_{r}\}_{r\in[t,T]}$ an admissible reinsurance strategy on $[t,T]$. We denote the set of all admissible reinsurance strategies on $[t, T]$ by $\mathcal{D}_{[t, T]}$.
\end{definition}

For each reinsurer's premium strategy  $c\in\mathcal{C}$, the insurer develops a reinsurance strategy $\xi=\xi(c)\in\mathcal{D}_{[0,T]}$. Knowing the insurer's rule $c\mapsto\xi(c)$ to develop a reinsurance strategy, the reinsurer chooses a  premium strategy $c\in\mathcal{C}$ to minimize
\begin{equation*}
J^{R}(x,t,y,z;c):=\mathbb{E}_{z,t}Z^{\xi}_{T}+\gamma^{R}\var_{z,t}Z^{\xi}_{T}-\theta^{R}\mathbb{E}_{z,t}\int_{t}^{T}e^{\rho^{R}(T-r)}c(r)d\xi_{r}
\end{equation*}
with $\xi=\xi(c)$. The positive constants $\gamma^{R} $ and $\theta^{R}$ balance the three objective terms, which reflects the reinsurer's preferences on variance and reinsurance. The constant $\rho^{R}$ is the reinsurer's time discount factor. Because the process $\xi$ depends on the initial $(x,t,y)$, $J^{R}$ also depends on  $(x,t,y)$.

\subsection{Equilibrium strategies}
\label{subsec-es}
Both the insurer's and reinsurer's problems are time-inconsistent due to the variance term in the objectives. This subsection introduces the definitions of equilibrium strategies of the insurer and the reinsurer.

The insurer's problem is a time-inconsistent singular control problem. First, we define the insurer's admissible reinsurance law.

\begin{definition}[the insurer's admissible reinsurance law]
\label{adl}
Let $\Xi:=(W^{\Xi},P^{\Xi})$ be a division of $\mathcal{Q}$. We call $\Xi$ an 
 insurer's admissible reinsurance law if given any initial $(x,t,y)\in\mathcal{Q}$, the following Skorohod reflection problem
	\begin{equation}
	\left\{
	\begin{array}{l}
dX^{\xi}_{r}=(a^{I}-b^{I}X^{\xi}_{r})dr+\sigma^{I}dB^{I}_{r}-d\xi_{r},\ r\in[t,T],\\
	(X^{\xi}_{r},r,\xi_{r})\in\overline{W^{\Xi}}, r\in[t,T],\\
	\xi_{r}=\int_{t}^{r}1_{\{(X^{\xi}_{u},u,\xi_{u})\in P^{\Xi}\}}d\xi_{u},\ r\in[t,T],\\ 
	X^{\xi}_{t-}=x,\ \xi_{t-}=y,
	\end{array}
	\right.
	\label{dynamic1}
	\end{equation}
	has a unique strong solution $(X^{\xi},\xi):=(X^{x,t,y,\Xi},\xi^{x,t,y,\Xi})$  in $[t,T)$ and (b) of Definition \ref{ads} is satisfied. We call $\xi$ a reinsurance strategy generated by $\Xi$ at $(x,t,y)$.
 \end{definition}
Then we define the insurer's equilibrium reinsurance law and the equilibrium reinsurance strategy as follows.

\begin{definition}[the insurer's equilibrium reinsurance law; equilibrium reinsurance strategy]
\label{eqli}
 Suppose  that $\hat{\Xi}$ is an insurer's admissible reinsurance law, we call $\hat{\Xi}$ an insurer's equilibrium reinsurance law if the following (a) and (b) hold.\\
 (a) For any initial $(x,t,y)\in\mathcal{Q}$ with $t<T$, any $\eta\in\mathcal{D}_{[t,T]}$ with $\eta_{t-}=y$, we have
	\begin{displaymath}
	\liminf\limits_{h\rightarrow 0}\frac{J^{I}(x,t,y;\xi^{h})-J^{I}(x,t,y;\xi^{x,t,y,\hat{\Xi}})}{h}\ge 0,
	\end{displaymath}
	where $\xi^{h}, h\in(0,T-t)$, is defined by
	\begin{equation*}
	\xi^{h}_{r}=\left\{
	\begin{array}{l}
	\eta_{r}, \ r\in[t,t+h),\\
	\xi^{X^{\eta}_{(t+h)-},t+h,\eta_{(t+h)-},\hat{\Xi}}_{r}, \ r\in[t+h,T].
	\end{array}
	\right.
	\end{equation*}
	(b) For any initial $(x,t,y)\in\mathcal{Q}$ with $t=T$, any $\eta\in\mathcal{D}_{[T,T]}$ with $\eta_{T-}=y$, we have
	\begin{displaymath}
	J^{I}(x,t,y;\eta)-J^{I}(x,t,y;\xi^{x,t,y,\hat{\Xi}})\ge 0.
	\end{displaymath}
 
 We call $J^{I}(x,t,y;\xi^{x,t,y,\hat{\Xi}})$ the corresponding equilibrium value function of $\hat{\Xi}$. Any reinsurance strategy generated by an insurer's equilibrium reinsurance law $\hat{\Xi}$ defined (\ref{dynamic1}) is called an equilibrium reinsurance strategy.
\end{definition}

Now let us introduce the equilibrium of the reinsurer's problem. The reinsurer's problem seems to be a time-inconsistent regular control problem at first glance. We will see that there are many features similar to singular control problems or stopping problems, e.g., partly changing the control value may lead to no change in the objective value; changing the control's value at one point may lead to a change in the objective value. The main reason is that the insurer's problem is a singular control problem and the singularity features can be transferred to the reinsurer's problem through the equilibrium control.

Further, we will show that the reinsurer's problem can be transformed into a time-inconsistent time-selection problem, which is similar to time-inconsistent stopping (see \cite{bjork2021time}). Moreover, the reinsurer's equilibrium corresponds to the equilibrium of the time-inconsistent time-selection problem, which can be defined in the same way as time-inconsistent stopping problems. 

The following definition follows the standard perturbation procedure in defining a time inconsistent equilibrium. We emphasize here that the terminal optimality condition (b) is required due to the transferred singular structure.

\begin{definition}[equilibrium premium strategy]
\label{eqps}
We call $\hat{c}\in\mathcal{C}$ an equilibrium premium strategy if the following (a) and (b) hold.\\
(a) For any initial $(x,t,y,z)\in\mathcal{Q}\times\mathbb{R}$ with $t<T$ and  $\tilde{c}\in\mathcal{C}$, we have
	\begin{displaymath}
	\liminf\limits_{h\rightarrow 0}\frac{J^{R}(x,t,y,z;c^{h})-J^{R}(x,t,y,z;\hat{c})}{h}\ge 0,
	\end{displaymath}
	where $c^{h}, h\in(0,T-t)$, is defined by
	\begin{equation*}
	c^{h}_{r}=\left\{
	\begin{array}{l}
	\tilde{c}_{r}, \ r\in[t,t+h),\\
	\hat{c}_{r}, \ r\in[t+h,T],
	\end{array}
	\right.
	\end{equation*}
 and $\xi$ in the expression of $J^{R}(x,t,y;\cdot)$ is given by $\xi=\xi^{x,t,y;\hat{\Xi}}(\cdot)$ with $\hat{\Xi}$ being an insurer's equilibrium reinsurance law.\\
 (b) For any initial $(x,t,y,z)\in\mathcal{Q}\times\mathbb{R}$ with $t=T$ and $\tilde{c}\in\mathcal{C}$, we have
 \begin{equation*}
 J^{R}(x,t,y,z;\tilde{c})-J^{R}(x,t,y,z;\hat{c})\ge 0.
 \end{equation*}
\end{definition}

\subsection{An assumption}

We end this section with an assumption for the rest of the paper. Readers may just accept the assumption and skip the explanation behind it at first reading.

\begin{assumption}
\label{assump}
We assume that the insurer would purchase as much reinsurance as possible if purchasing reinsurance makes no difference to her overall objective $J^{I}$.
\end{assumption}

The most direct intention of Assumption \ref{assump} is to make the state path $X^{x,t,y,\Xi}$ and control path $\xi^{x,t,y,\Xi}$ determined at terminal time $T$, see Definition \ref{adl}. Hereafter, we refer $X^{x,t,y,\Xi}$ and $\xi^{x,t,y,\Xi}$ to the processes uniquely determined by (\ref{dynamic1}) on $[t, T]$ under Assumption \ref{assump}.

There might be other alternative assumptions that determine the insurer's behavior in such a case, e.g., she purchases half the largest possible amount or does not purchase at all. However, we adopt Assumption \ref{assump} in the sense that it makes our analysis easiest to understand.

For the insurer's problem in Section \ref{sec-insurer}, Assumption \ref{assump} merges the terminal boundary region where $\theta^{I}c(T)=1$ into the purchasing region. With some other assumptions, the purchasing region can be shrunk as the jump of equilibrium reinsurance strategy might be smaller. In an extreme case, if we assume  that the insurer does not purchase at all in the situation, then even the whole terminal boundary region where $\theta^{I}c(T)=1$ can be merged into the waiting region. Other than that, with or without the assumption, or with which assumption, does not have an impact on the insurer's equilibrium reinsurance law or equilibrium reinsurance strategy before $T$. In fact, we can use the alternative auxiliary function $\tilde{g}(x,t,y)=E_{x,t,y}X^{\hat{\xi}_{T-}}$ in an alternative verification theorem for the insurer's equilibrium reinsurance law before $T$. However, posing an assumption determines the terminal value of the auxiliary function $g$, which enables the verification theorem we use. 

For the reinsurer's problem in Section \ref{sec-reinsurer}, Assumption \ref{assump} determines the amount of risk transferred to the reinsurer if the purchasing time $p(t;c)$ equals $T$, which ensures the well-posedness of the reinsurer's problem. The amount of risk transferred is the same as that of the case $p(t;c)<T$ under Assumption \ref{assump}. Other alternative assumptions can also ensure the well-posedness of the reinsurer's problem, but the case $p(t;c)=T$ corresponds to a different amount of risk transferred. However, the analysis is expected to be similar and we just consider the case under Assumption \ref{assump}.

\section{The Insurer's Problem}
\label{sec-insurer}

This section is dedicated to solving the insurer's problem. The insurer's problem is a time-inconsistent singular control problem. To obtain an equilibrium reinsurance strategy for the insurer, we only need to obtain the insurer's equilibrium reinsurance law, which generates equilibrium reinsurance strategies through (\ref{dynamic1}).

Assuming enough regularity, we obtain the extended HJB equations for the insurer's equilibrium reinsurance law and corresponding equilibrium value function in Subsection \ref{subsec-hjb}. Then, in Subsection \ref{subsec-verif}, we obtain a solution to the HJB equations and verify that it indeed leads to an insurer's equilibrium reinsurance law.

\subsection{Extended HJB equations from a verification theorem}
\label{subsec-hjb}
In this subsection, we derive the extended HJB equations of the insurer's problem by proposing a verification theorem.

The following verification theorem provides the extended HJB equations under enough regularity conditions. 

\begin{theorem}[verification theorem of the insurer's problem]
\label{verif-insurer}
Given two functions $V^{I}(x,t,y)$ and $g(x,t,y)$, define the infinitesimal operator $\mathcal{A}$ by
\begin{equation*}
(\mathcal{A}\varphi)(x,t,y):=\varphi_{t}(x,t,y)+(a^{I}-b^{I}x)\varphi_{x}(x,t,y)+\frac{1}{2}(\sigma^{I})^{2}\varphi_{xx}(x,t,y),\forall\varphi      \in C^{2,1,1}(\mathcal{Q}),
%:\mathcal{Q}\rightarrow\mathbb{R},
\end{equation*}
and the $\hat{\Xi}=(W^{\hat{\Xi}},P^{\hat{\Xi}})$ by
	\begin{align}
	W^{\hat{\Xi}}:=&\big\{(x,t,y)\in \mathcal{Q}\big|\theta^{I}e^{\rho^{I}(T-t)}c(t)-V^{I}_{x}(x,t,y)+V^{I}_{y}(x,t,y)>0\ {\rm or}\ y=\bar{y}\big\},\label{W}\\
	P^{\hat{\Xi}}:=&\big\{(x,t,y)\in \mathcal{Q}\big|\theta^{I}e^{\rho^{I}(T-t)}c(t)-V^{I}_{x}(x,t,y)+V^{I}_{y}(x,t,y)=0,y<\bar{y}\big\}.\label{P}
	\end{align}
 
	Assume the following properties:\\
	(1) 
	\begin{equation*}
		V^{I},g\in C^{2,1,1}(\mathcal{Q}).
	\end{equation*}
	(2) $V^{I}(x,t,y)$ and $g(x,t,y)$ satisfy \ $ \forall t<T $
	\begin{align}
	&\min\Big\{(\mathcal{A}V^{I})(x,t,y)+\frac{\gamma^{I}}{2}(\sigma^{I})^{2}g_{x}(x,t,y)^2,\theta^{I}e^{\rho^{I}(T-t)}c(t)-V^{I}_{x}(x,t,y)+V^{I}_{y}(x,t,y)\Big\}=0,\ \label{v1}\\
	&V^{I}(x,T,y)=x-(\bar{y}-y)\max\big\{1-\theta^{I}c(T),0\big\},\label{v2}\\
	&(\mathcal{A}g)(x,t,y)=0, \forall (x,t,y)\in \overline{W^{\hat{\Xi}}}, \ \label{v3}\\
	&g_{x}(x,t,y)-g_{y}(x,t,y)=0,\forall (x,t,y)\in P^{\hat{\Xi}}, \ \label{v4}\\
	&g(x,T,y)=x-(\bar{y}-y)1_{\theta^{I}c(T)\le 1},\label{v5}
	\end{align}
	for all $(x,t,y)\in\mathcal{Q}$, where $1_{A}$ in (\ref{v5}) is the indicator function that takes value $1$ if  $A$ holds and takes value $0$ otherwise.\\
	(3) $\hat{\Xi}$ is an insurer's admissible reinsurance law.\\
	(4) For $\varphi=V^{I},g,g^{2}$ and for any $\xi\in\mathcal{D}_{[t,T]}$,
	\begin{equation*}
	\mathbb{E}_{x,t,y}\int_{t}^{T}\varphi_{x}(X_{r}^{\xi},r,\xi_{r})^{2}dr<\infty.
	\end{equation*}
	Then $\hat{\Xi}$ is an insurer's equilibrium insurance law and $V^{I}$ is the corresponding equilibrium value function of $\hat{\Xi}$. Moreover, $g$ has the probabilistic interpretation
	\begin{equation}
	\label{intg}
g(x,t,y)=E_{x,t,y}X^{\hat{\xi}}_{T},
	\end{equation}
	where $\hat{\xi}:=\xi^{x,t,y,\hat{\Xi}}$ is the singular control generated by $\hat{\Xi}$ at $(x,t,y)$ and $X^{\hat{\xi}}:=X^{x,t,y,\hat{\Xi}}$ is the corresponding risk exposure process.
\end{theorem}

\begin{proof}
{\bf Step 1:} We show that $g$ has the probabilistic interpretation (\ref{intg}).
Applying It\^{o}-Tanaka-Meyer's formula to $g$ on $[t,T)$, we obtain 
\begin{align*}
&g(X^{\hat{\xi}}_{T-},T,\hat{\xi}_{T-})-g(x,t,y)\\
=&\int_{t}^{T}(\mathcal{A}g)(X^{\hat{\xi}}_{r-},r,\hat{\xi}_{r-})dr+\int_{t}^{T}\sigma^{I} g_{x}(X^{\hat{\xi}}_{r-},r,\hat{\xi}_{r-})dB^{I}\\
&+\int_{t}^{T}\Big[g_{y}(X^{\hat{\xi}}_{r-},r,\hat{\xi}_{r-})-g_{x}(X^{\hat{\xi}}_{r-},r,\hat{\xi}_{r-})\Big]d\hat{\xi}^{c}_{r}\\
&+\sum\limits_{r\in[t,T)}\int_{0}^{\Delta\hat{\xi}_{r}}\Big[g_{y}(X^{\hat{\xi}}_{r-}-u,r,\hat{\xi}_{r-}+u)-g_{x}(X^{\hat{\xi}}_{r-}-u,r,\hat{\xi}_{r-}+u)\Big]du.
\end{align*}
The integration w.r.t $dB^{I}$ on the right-hand side is a martingale under Condition (4). Then, taking expectations on both sides and using Conditions (\ref{v3}), (\ref{v4}) and (3), we obtain 
\begin{align*}
&\mathbb{E}_{x,t,y}g(X^{\hat{\xi}}_{T-},T,\hat{\xi}_{T-})-g(x,t,y)\\
=&\mathbb{E}_{x,t,y}\int_{t}^{T}(\mathcal{A}g)(X^{\hat{\xi}}_{r-},r,\hat{\xi}_{r-})dr+\mathbb{E}_{x,t,y}\int_{t}^{T}\Big[g_{y}(X^{\hat{\xi}}_{r-},r,\hat{\xi}_{r-})-g_{x}(X^{\hat{\xi}}_{r-},r,\hat{\xi}_{r-})\Big]d\hat{\xi}^{c}_{r}\\
&+\mathbb{E}_{x,t,y}\sum\limits_{r\in[t,T)}\int_{0}^{\Delta\hat{\xi}_{r}}\Big[g_{y}(X^{\hat{\xi}}_{r-}-u,r,\hat{\xi}_{r-}+u)-g_{x}(X^{\hat{\xi}}_{r-}-u,r,\hat{\xi}_{r-}+u)\Big]du=0.
\end{align*}Using (\ref{v5}) yields 
\begin{align*}
g(x,t,y)=&\mathbb{E}_{x,t,y}g(X^{\hat{\xi}}_{T-},T,\hat{\xi}_{T-})
=\mathbb{E}_{x,t,y}\Big[X^{\hat{\xi}}_{T-}-(\bar{y}-\hat{\xi}_{T-})1_{\theta^{I}c(T)\le 1}\Big]
=\mathbb{E}_{x,t,y}X^{\hat{\xi}}_{T}.
\end{align*}
Then we have the probabilistic interpretation (\ref{intg}).

{\bf Step 2:} We prove $V^{I}(x,t,y)=J^{I}(x,t,y;\hat{\xi})$. Observe that
\begin{equation*}
\mathcal{A}(g^{2})(x,t,y)-2g(x,t,y)(\mathcal{A}g(x,t,y))=\frac{1}{2}(\sigma^{I})^{2}(x,t,y)g_{x}^{2}(x,t,y),
\end{equation*}
then, using Conditions (\ref{v1}) and (\ref{v3}), we have
\begin{equation}
\label{vt1}
(\mathcal{A}V^{I})(x,t,y)=-\gamma^{I}\mathcal{A}(g^{2})(x,t,y),\forall (x,t,y)\in\overline{W^{\hat{\Xi}}}.
\end{equation}
Similar to {\bf Step 1}, applying It\^{o}-Tanaka-Meyer's formula to $V^{I}$ on $[t, T]$ and taking expectations yield 
\begin{align}
&\mathbb{E}_{x,t,y}V^{I}(X^{\hat{\xi}}_{T},T,\hat{\xi}_{T})-V^{I}(x,t,y)\label{vt2}\\
=&\mathbb{E}_{x,t,y}\int_{t}^{T}(\mathcal{A}V^{I})(X^{\hat{\xi}}_{r-},r,\hat{\xi}_{r-})dr+\mathbb{E}_{x,t,y}\int_{t}^{T}\Big[V^{I}_{y}(X^{\hat{\xi}}_{r-},r,\hat{\xi}_{r-})-V^{I}_{x}(X^{\hat{\xi}}_{r-},r,\hat{\xi}_{r-})\Big]d\hat{\xi}^{c}_{r}\notag\\
&+\mathbb{E}_{x,t,y}\sum\limits_{r\in[t,T]}\int_{0}^{\Delta\hat{\xi}_{r}}\Big[V^{I}_{y}(X^{\hat{\xi}}_{r-}-u,r,\hat{\xi}_{r-}+u)-V^{I}_{x}(X^{\hat{\xi}}_{r-}-u,r,\hat{\xi}_{r-}+u)\Big]du\notag\\
=&-\gamma^{I}\mathbb{E}_{x,t,y}\int_{t}^{T}(\mathcal{A}g^{2})(X^{\hat{\xi}}_{r-},r,\hat{\xi}_{r-})dr-\theta^{I}\mathbb{E}_{x,t,y}\int_{t}^{T}e^{\rho^{I}(T-r)}c(r)d\hat{\xi}^{c}_{r}\notag\\
&-\theta^{I}\mathbb{E}_{x,t,y}\sum\limits_{r\in[t,T]}e^{\rho^{I}(T-r)}c(r)\Delta\hat{\xi}_{r}\notag\\
=&-\gamma^{I}\mathbb{E}_{x,t,y}\int_{t}^{T}(\mathcal{A}g^{2})(X^{\hat{\xi}}_{r-},r,\hat{\xi}_{r-})dr-\theta^{I}\mathbb{E}_{x,t,y}\int_{t}^{T}e^{\rho^{I}(T-r)}c(r)d\hat{\xi}_{r},\notag
\end{align}
where the second equality follows from Conditions (\ref{W}), (\ref{P}), (\ref{v1}) and Eq.(\ref{vt1}),  and the last equality is trivial.
Moreover,
\begin{align}
&\mathbb{E}_{x,t,y}(X^{\hat{\xi}}_{T})^{2}-(\mathbb{E}_{x,t,y}X^{\hat{\xi}}_{T})^{2}\label{vt3}\\
=&\mathbb{E}_{x,t,y}g^{2}(X^{\hat{\xi}}_{T},T,\hat{\xi}_{T})-g^{2}(x,t,y)\notag\\
=&\mathbb{E}_{x,t,y}\int_{t}^{T}(\mathcal{A}g^2)(X^{\hat{\xi}}_{r-},r,\hat{\xi}_{r-})dr\!+\!\mathbb{E}_{x,t,y}\!\int_{t}^{T}\!\Big[(2gg_{y})(X^{\hat{\xi}}_{r-},r,\hat{\xi}_{r-})\!-\!(2gg_{x})(X^{\hat{\xi}}_{r-},r,\hat{\xi}_{r-})\Big]d\hat{\xi}^{c}_{r}\notag\\
&+\sum\limits_{r\in[t,T]}\mathbb{E}_{x,t,y}\int_{0}^{\Delta\hat{\xi}_{r}}\Big[(2gg_{y})(X^{\hat{\xi}}_{r-}\!-\!u,r,\hat{\xi}_{r-}\!+\!u)\!-\!(2gg_{x})(X^{\hat{\xi}}_{r-}\!-\!u,r,\hat{\xi}_{r-}\!+\!u)\Big]du\notag\\
=&\mathbb{E}_{x,t,y}\int_{t}^{T}(\mathcal{A}g^2)(X^{\hat{\xi}}_{r-},r,\hat{\xi}_{r-})dr,\notag
\end{align}
where the first equality follows from (\ref{intg}), the second equality follows from applying It\^{o}-Tanaka-Meyer's formula to $g^{2}$ on $[t, T]$ and taking expectations, and the last equality follows from Condition (\ref{v4}).

Substituting (\ref{vt3}) and (\ref{v2}) into (\ref{vt2}), we obtain
\begin{equation*}
V^{I}(x,t,y)\!=\!\mathbb{E}_{x,t,y}X^{\hat{\xi}}_{T}+\gamma^{I}\var_{x,t,y}X^{\hat{\xi}}_{T}+\theta^{I}\mathbb{E}_{x,t,y}\int_{t}^{T}e^{\rho^{I}(T-r)}c(r)d\hat{\xi}_{r}-(\bar{y}-\hat{\xi}_{T})\max\{1-\theta^{I}c(T),0\}.
\end{equation*}

Observe that if $1-\theta^{I}c(T)>0$ then all reinsurance must be purchased at $T$ and $\hat{\xi}_{T}=\bar{y}$, thus
\begin{equation}
\label{vt4}
V^{I}(x,t,y)=J^{I}(x,t,y;\hat{\xi}).
\end{equation}

{\bf Step 3:} We prove that $\hat{\Xi}$ is the insurer's equilibrium reinsurance law and $V^{I}$ is the corresponding equilibrium value function of $\hat{\Xi}$.

By (\ref{v2}), $\hat{\Xi}$ satisfies Condition (b) of Definition \ref{eqli}. We only prove that $\hat{\Xi}$ satisfies Condition (a) of Definition \ref{eqli}.
Using Definitions of $J^{I}$ and $\xi^{h}$, we have
\begin{align*}
J^{I}(x,t,y;\xi^{h})=&\mathbb{E}_{x,t,y}J^{I}(X^{\eta}_{(t+h)-},t+h,\eta_{(t+h)-};\hat{\xi})+\gamma^{I}\mathbb{E}_{x,t,y}g^{2}(X^{\eta}_{(t+h)-},t+h,\eta_{(t+h)-})\\
&-\gamma^{I}(\mathbb{E}_{x,t,y}g(X^{\eta}_{(t+h)-},t+h,\eta_{(t+h)-}))^{2}+\theta^{I}\mathbb{E}_{x,t,y}\int_{t}^{t+h}e^{\rho^{I}(T-r)}c(r)d\eta^{c}_{r}\\
&+\theta^{I}\mathbb{E}_{x,t,y}\sum\limits_{r\in[t,t+h)}e^{\rho^{I}(T-r)}c(r)\Delta\eta_{r}.
\end{align*} 
Then, using (\ref{vt4}), we have 
\begin{align}
&J^{I}(x,t,y;\xi^{h})-J^{I}(x,t,y;\hat{\xi})\label{vt5}\\
=&\mathbb{E}_{x,t,y}V^{I}(X^{\eta}_{(t+h)-},t+h,\eta_{(t+h)-})-V^{I}(x,t,y)\notag\\
&+\gamma^{I}\Big[\mathbb{E}_{x,t,y}g^{2}(X^{\eta}_{(t+h)-},t+h,\eta_{(t+h)-})-g^{2}(x-\Delta\eta_{t},t,y+\Delta\eta_{t})\Big]\notag\\
&-\gamma^{I}\Big[(\mathbb{E}_{x,t,y}g(X^{\eta}_{(t+h)-},t+h,\eta_{(t+h)-}))^{2}-g^{2}(x-\Delta\eta_{t},t,y+\Delta\eta_{t})\Big]\notag\\
&+\gamma^{I}\var_{x,t,y}g(X^{\eta}_{t},t,\eta_{t})\notag\\
&+\mathbb{E}_{x,t,y}\int_{t}^{t+h}\theta^{I}e^{\rho^{I}(T-r)}c(r)d\eta^{c}_{r}+\mathbb{E}_{x,t,y}\theta^{I}e^{\rho^{I}(T-t)}c(t)\Delta\eta_{t}.\notag
\end{align} 
Applying It\^{o}-Tanaka-Meyer's formula to $V^{I}$, $g^{2}$ and $g$, we have the following estimates.
\begin{align}
&\mathbb{E}_{x,t,y}V^{I}(X^{\eta}_{(t+h)-},t+h,\eta_{(t+h)-})-V^{I}(x,t,y)\label{vt6}\\
=&\mathbb{E}_{x,t,y}\int_{t}^{t+h}(\mathcal{A}V^{I})(X^{\eta}_{r-},r,\eta_{r-})dr+\mathbb{E}_{x,t,y}\int_{t}^{t+h}\Big[V^{I}_{y}(X^{\eta}_{r-},r,\eta_{r-})-V^{I}_{x}(X^{\eta}_{r-},r,\eta_{r-})\Big]d\eta^{c}_{r}\notag\\
&+\mathbb{E}_{x,t,y}\sum\limits_{r\in[t,t+h)}\int_{0}^{\Delta\eta_{r}}\Big[V^{I}_{y}(X^{\eta}_{r-}-u,r,\eta_{r-}+u)-V^{I}_{x}(X^{\eta}_{r-}-u,r,\eta_{r-}+u)\Big]du,\notag
\end{align}
\begin{align}
&\mathbb{E}_{x,t,y}g^{2}(X^{\eta}_{(t+h)-},t+h,\eta_{(t+h)-})-g^{2}(x-\Delta\eta_{t},t,y+\Delta\eta_{t})\label{vt7}\\
=&\mathbb{E}_{x,t,y}\int_{t}^{t+h}(\mathcal{A}g^{2})(X^{\eta}_{r-},r,\eta_{r-})dr\!+\!\mathbb{E}_{x,t,y}\int_{t}^{t+h}\Big[(2gg_{y})(X^{\eta}_{r-},r,\eta_{r-})\!-\!(2gg_{x})(X^{\eta}_{r-},r,\eta_{r-})\Big]d\eta^{c}_{r}\notag\\
&+\mathbb{E}_{x,t,y}\sum\limits_{r\in(t,t+h)}\int_{0}^{\Delta\eta_{r}}\Big[(2gg_{y})(X^{\eta}_{r-}-u,r,\eta_{r-}+u)-(2gg_{x})(X^{\eta}_{r-}-u,r,\eta_{r-}+u)\Big]du,\notag
\end{align}
\begin{align}
&(\mathbb{E}_{x,t,y}g(X^{\eta}_{(t+h)-},t+h,\eta_{(t+h)-}))^{2}-g^{2}(x-\Delta\eta_{t},t,y+\Delta\eta_{t})\label{vt8}\\
=&2\Big[\mathbb{E}_{x,t,y}g(X^{\eta}_{(t+h)-},t+h,\eta_{(t+h)-})-g(x-\Delta\eta_{t},t,y+\Delta\eta_{t})\Big]g(x-\Delta\eta_{t},t,y+\Delta\eta_{t})\notag\\
&+o\Big(\mathbb{E}_{x,t,y}g(X^{\eta}_{(t+h)-},t+h,\eta_{(t+h)-})-g(x-\Delta\eta_{t},t,y+\Delta\eta_{t})\Big)\notag\\
=&2g(x\!-\!\Delta\eta_{t},t,y\!+\!\Delta\eta_{t})\mathbb{E}_{x,t,y}\bigg\{\int_{t}^{t+h}(\mathcal{A}g)(X^{\eta}_{r-},r,\eta_{r-})dr+\int_{t}^{t+h}\Big[g_{y}(X^{\eta}_{r-},r,\eta_{r-})\notag\\&-g_{x}(X^{\eta}_{r-},r,\eta_{r-})\Big]d\eta^{c}_{r}
+\sum\limits_{r\in(t,t+h)}\int_{0}^{\Delta\eta_{r}}\Big[g_{y}(X^{\eta}_{r-}\!-\!u,r,\eta_{r-}\!+\!u)\!-\!g_{x}(X^{\eta}_{r-}\!-\!u,r,\eta_{r-}\!+\!u)\Big]du\bigg\}\notag\\
&+o(h).\notag
\end{align}
Substituting (\ref{vt6})$\sim$(\ref{vt8}) into (\ref{vt5}), we have
\begin{align*}
&J^{I}(x,t,y;\xi^{h})-J^{I}(x,t,y;\hat{\xi})\\
=&\mathbb{E}_{x,t,y}\int_{t}^{t+h}\Big[(\mathcal{A}V^{I})(X^{\eta}_{r-},r,\eta_{r-})+\gamma^{I}(\mathcal{A}g^{2})(X^{\eta}_{r-},r,\eta_{r-})\Big]dr\\
&+\mathbb{E}_{x,t,y}\int_{t}^{t+h}\Big[\theta^{I}e^{\rho^{I}(T-r)}c(r)+V^{I}_{y}(X^{\eta}_{r-},r,\eta_{r-})-V^{I}_{x}(X^{\eta}_{r-},r,\eta_{r-})\Big]d\eta^{c}_{r}\\
&+\mathbb{E}_{x,t,y}\sum\limits_{r\in[t,t+h)}\int_{0}^{\Delta\eta_{r}}\Big[\theta^{I}e^{\rho^{I}(T-r)}c(r)\!+\!V^{I}_{y}(X^{\eta}_{r-}\!-\!u,r,\eta_{r-}\!+\!u)\!-\!V^{I}_{x}(X^{\eta}_{r-}\!-\!u,r,\eta_{r-}\!+\!u)\Big]du\\
&+\gamma^{I}\mathbb{E}_{x,t,y}\int_{t}^{t+h}\Big[(2gg_{y})(X^{\eta}_{r-},r,\eta_{r-})-(2gg_{x})(X^{\eta}_{r-},r,\eta_{r-})\Big]d\eta^{c}_{r}\\
&+\gamma^{I}\mathbb{E}_{x,t,y}\sum\limits_{r\in(t,t+h)}\int_{0}^{\Delta\eta_{r}}\Big[(2gg_{y})(X^{\eta}_{r-}-u,r,\eta_{r-}+u)-(2gg_{x})(X^{\eta}_{r-}-u,r,\eta_{r-}+u)\Big]du\\
&-2\gamma^{I}g(x\!-\!\Delta\eta_{t},t,y\!+\!\Delta\eta_{t})\mathbb{E}_{x,t,y}\bigg\{\int_{t}^{t+h}(\mathcal{A}g)(X^{\eta}_{r-},r,\eta_{r-})dr+\int_{t}^{t+h}\Big[g_{y}(X^{\eta}_{r-},r,\eta_{r-})\\
&-g_{x}(X^{\eta}_{r-},r,\eta_{r-})\Big]d\eta^{c}_{r}\!+\!\sum\limits_{r\in(t,t+h)}\!\int_{0}^{\Delta\eta_{r}}\!\Big[g_{y}(X^{\eta}_{r-}\!-\!u,r,\eta_{r-}\!+\!u)\!-\!g_{x}(X^{\eta}_{r-}\!-\!u,r,\eta_{r-}\!+\!u)\Big]du\bigg\}
\\
&+\gamma^{I}\var_{x,t,y}g(X^{\eta}_{t},t,\eta_{t})+o(h)
\\
\ge&2\gamma^{I}\mathbb{E}_{x,t,y}\bigg\{\int_{t}^{t+h}\Big[g(X^{\eta}_{r-},r,\eta_{r-})-g(x\!-\!\Delta\eta_{t},t,y\!+\!\Delta\eta_{t})\Big](\mathcal{A}g)(X^{\eta}_{r-},r,\eta_{r-})dr\\
&+\int_{t}^{t+h}\Big[g(X^{\eta}_{r-},r,\eta_{r-})\!-\!g(x\!-\!\Delta\eta_{t},t,y\!+\!\Delta\eta_{t})\Big]\!\Big[g_{y}(X^{\eta}_{r-},r,\eta_{r-})\!-\!g_{x}(X^{\eta}_{r-},r,\eta_{r-})\Big]d\eta^{c}_{r}\\
&+\!\sum\limits_{r\in(t,t+h)}\!\int_{0}^{\Delta\eta_{r}}\!\Big[g(X^{\eta}_{r-}\!-\!u,r,\eta_{r-}\!+\!u)\!-\!g(x\!-\!\Delta\eta_{t},t,y\!+\!\Delta\eta_{t})\Big]\!\Big[g_{y}(X^{\eta}_{r-}\!-\!u,r,\eta_{r-}\!+\!u)\\
&\!-\!g_{x}(X^{\eta}_{r-}\!-\!u,r,\eta_{r-}\!+\!u)\Big]du\bigg\}\\
&+o(h),
\end{align*}
where the inequality follows from Conditions (\ref{v1}) and (3)  as well as the non-negativity of variance. Then, using the dominated convergence theorem, we have
\begin{equation*}
\liminf\limits_{h\rightarrow 0}\frac{J^{I}(x,t,y;\xi^{h})-J^{I}(x,t,y;\hat{\xi})}{h}\ge 0.
\end{equation*}
Thus the proof is complete.
\end{proof}

Now we manage to find a solution to extended HJB equations (\ref{v1})$\sim$(\ref{v5}).

\subsection{ A solution to the extended HJB equations}
\label{subsec-verif}
In this subsection, we obtain a solution to the extended HJB equations by making an ansatz of the solution structure. Though the solution does not satisfy the regularity condition (1) of Theorem \ref{verif-insurer}, we show that the verification proof is still valid and we indeed obtain an insurer's equilibrium reinsurance law.

Note that the optimal strategy at $T$ either purchases all possible reinsurance or does not purchase at all, depending on whether $\theta^{I}c(T)\le 1$. We conjecture that the division $(W^{\hat{\Xi}},P^{\hat{\Xi}})$ constrained on $y<\bar{y}$ depends only on time. We make an ansatz that
\begin{align*}
&V^{I}(x,t,y)=u_{1}(t)x+u_{2}(t)y+u_{3}(t),\\
&g(x,t,y)=v_{1}(t)x+v_{2}(t)y+v_{3}(t).
\end{align*}
Then the division $\hat{\Xi}=(W^{\hat{\Xi}},P^{\hat{\Xi}})$ given by (\ref{W}) and (\ref{P}) becomes
\begin{align*}
&W^{\hat{\Xi}}=\Big\{(x,t,y)\in\mathcal{Q}\Big|\theta^{I}e^{\rho^{I}(T-t)}c(t)-u_{1}(t)+u_{2}(t)>0\ {\rm or}\ y=\bar{y}\Big\},\\
&P^{\hat{\Xi}}=\Big\{(x,t,y)\in\mathcal{Q}\Big|\theta^{I}e^{\rho^{I}(T-t)}c(t)-u_{1}(t)+u_{2}(t)=0,y<\bar{y}\Big\}.
\end{align*}
The extended HJB equations (\ref{v1})$\sim$(\ref{v5}) become
\begin{align}
&\min\left\{u_{1}'(t)x+u_{2}'(t)y+u_{3}'(t)+(a^{I}-b^{I}x)u_{1}(t)+\frac{\gamma^{I}}{2}(\sigma^{I})^{2}v_{1}(t)^{2}\right. ,\label{vs1}\\
&\qquad\quad\left.\theta^{I}e^{\rho^{I}(T-t)}c(t)-u_{1}(t)+u_{2}(t)\right \}=0,\notag\\
&v_{1}'(t)x+v_{2}'(t)y+v_{3}'(t)+(a^{I}-b^{I}x)v_{1}(t)=0,\forall (x,t,y)\in \overline{W^{\hat{\Xi}}},\label{vs2}\\
&v_{1}(t)-v_{2}(t)=0, \forall (x,t,y)\in P^{\hat{\Xi}}\label{vs3}
\end{align}
with terminal conditions
\begin{align}
&u_{1}(T)=1,\label{vs4}\\
&u_{2}(T)=\max\{1-\theta^{I}c(T),0\},\label{vs5}\\
&u_{3}(T)=-\bar{y}\max\{1-\theta^{I}c(T),0\},\label{vs6}\\
&v_{1}(T)=1,\label{vs7}\\
&v_{2}(T)=1_{\theta^{I}c(T)\le 1},\label{vs8}\\
&v_{3}(T)=-\bar{y}1_{\theta^{I}c(T)\le 1}.\label{vs9}
\end{align}
Using (\ref{vs4}), (\ref{vs7}) and the arbitrariness of $x$ in (\ref{vs1}) and (\ref{vs2}), we have
\begin{align*}
&u_{1}(t)=e^{-b^{I}(T-t)},\\
&v_{1}(t)=e^{-b^{I}(T-t)}.
\end{align*}
We have used the fact that for any $(x,t,y)\in P^{\hat{\Xi}}$, $(x,t,\bar{y})$ is in $W^{\hat{\Xi}}$ to obtain the expression of $v_{1}$ on $P^{\hat{\Xi}}$. Then, from (\ref{vs1}),  for $t$ such that $u_{2}(t)>e^{-b^{I}(T-t)}-\theta^{I}e^{\rho^{I}(T-t)}c(t)$, we have
\begin{equation*}
u_{2}'(t)y+u_{3}'(t)+a^{I}e^{-b^{I}(T-t)}+\frac{\gamma^{I}}{2}(\sigma^{I})^{2}v_{1}(t)^{2}=0.
\end{equation*}
Using the arbitrariness of $y$ in the last equation and (\ref{vs5}) yields
\begin{equation*}
u_{2}(t)=\max\big\{e^{-b^{I}(T-t)}-\theta^{I}e^{\rho^{I}(T-t)}c(t),0\big\},
\end{equation*}
then we have that for $y<\bar{y}$, the point $(x,t,y)$ belongs to $W^{\hat{\Xi}}$ if and only if $c(t)>\kappa(t)$ where the threshold function $\kappa:[0,T]\mapsto\mathbb{R}$ is given by
\begin{equation}
\label{threshold}
\kappa(t):=\frac{1}{\theta^{I}}e^{-(b^{I}+\rho^{I})(T-t)},\forall t\in[0,T].
\end{equation}
To solve $v_{2}$, by (\ref{vs2}) and (\ref{vs3}), we have $v_{2}=v_{1}$ on $P^{\hat{\Xi}}$ and $v_{2}'=0$ on $\overline{W^{\hat{\Xi}}}$. In addition, we observe that for $(x,t,y)\in W^{\hat{\Xi}}$, the purchase takes place the first time the state-time-control triple enters the purchasing region. Hence the dependence on $y$ should be the same as that at the first entering time, i.e.,
\begin{equation*}
v_{2}(t)=e^{-b^{I}(T-p(t;c))},
\end{equation*}
where $p(t;c)$ is the first entering time (entering $P^{\hat{\Xi}}$) after $t$ given reinsurer's premium strategy $c(\cdot)$ and is set to infinity if there is no such time, i.e.,
\begin{equation}
\label{ptc}
p(t;c):=\left\{
\begin{array}{l}
\min\big\{\tau\in[t,T]\big|c(\tau)\le \kappa(\tau)\big\}, \ {\rm if\ such}\ \tau\ {\rm exists}.\\
+\infty,\ {\rm else}.
\end{array}
\right.
\end{equation}
The function $v_{2}$ solved above clearly satisfies the terminal condition (\ref{vs8}). To solve $u_{3}$, we have from (\ref{vs1}) that $u_{3}'(t)+a^{I}e^{-b^{I}(T-t)}+\frac{\gamma^{I}}{2}(\sigma^{I})^{2}e^{-2b^{I}(T-t)}=0$ on $\overline{W^{\hat{\Xi}}}$. In addition, we observe that for any $(x,t,y)\in P^{\hat{\Xi}}$, $V(x,t,y)=V(x+y-\bar{y},t,\bar{y})+\theta^{I}e^{\rho^{I}(T-t)}c(t)(\bar{y}-y)$ with $(x+y-\bar{y},t,\bar{y})\in W^{\hat{\Xi}}$. Denote
\begin{align*}
&u_{3}^{W^{\hat{\Xi}}}:=u_{3}\big|_{\{t|c(t)>\kappa(t)\}},\\
&u_{3}^{P^{\hat{\Xi}}}:=u_{3}\big|_{\{t|c(t)\le\kappa(t)\}}.
\end{align*}
Then the last observation implies that $u_{3}^{P^{\hat{\Xi}}}(t)=u_{3}^{W^{\hat{\Xi}}}(t)+[\theta^{I}e^{\rho^{I}(T-t)}c(t)-e^{-b^{I}(T-t)}]\bar{y}$. With terminal condition (\ref{vs6}), we obtain
\begin{equation*}
u_{3}(t)=-u_{2}(t)\bar{y}+\frac{a^{I}}{b^{I}}[1-e^{-b^{I}(T-t)}]+\frac{\gamma^{I}(\sigma^{I})^{2}}{4b^{I}}[1-e^{-2b^{I}(T-t)}].
\end{equation*}
Similarly, from (\ref{vs2}), we have $v_{3}'(t)+a^{I}e^{-b^{I}(T-t)}=0$. Using the terminal condition (\ref{vs9}) and the fact that $g(x,t,y)=g(x+y-\bar{y},t,\bar{y}),\forall(x,t,y)\in P^{\hat{\Xi}}$ with $(x+y-\bar{y},t,\bar{y})\in W^{\hat{\Xi}}$, we have
\begin{equation*}
v_{3}(t)=-v_{2}(t)\bar{y}+\frac{a^{I}}{b^{I}}[1-e^{-b^{I}(T-t)}].
\end{equation*}

Combining the above results gives a solution $(V^{I}, g)$ to the extended HJB equations. The solution $(V^{I}, g)$ is given by
\begin{align}
V^{I}(x,t,y)=&e^{-b^{I}(T-t)}x-\max\big\{e^{-b^{I}(T-t)}-\theta^{I}e^{\rho^{I}(T-t)}c(t),0\big\}(\bar{y}-y)\notag\\&+\frac{a^{I}}{b^{I}}\big[1-e^{-b^{I}(T-t)}\big]+\frac{\gamma^{I}(\sigma^{I})^{2}}{4b^{I}}\big[1-e^{-2b^{I}(T-t)}\big],\label{VI}\\
g(x,t,y)=&e^{-b^{I}(T-t)}x-e^{-b^{I}(T-p(t))}(\bar{y}-y)+\frac{a^{I}}{b^{I}}\big[1-e^{-b^{I}(T-t)}\big].\label{g}
\end{align}

We end this section with the following theorem, which collects the main results of this section.

\begin{theorem}
\label{theo-ins}
An insurer's equilibrium reinsurance law is $\hat{\Xi}=(W^{\hat{\Xi}},P^{\hat{\Xi}})$ where
\begin{align}
&W^{\hat{\Xi}}=\mathbb{R}\times\{t|c(t)>\kappa(t)\}\times[0,\bar{y})\bigcup\mathbb{R}\times[0,T]\times\{\bar{y}\},\label{W-prim}\\
&P^{\hat{\Xi}}=\mathbb{R}\times\{t|c(t)\le \kappa(t)\}\times[0,\bar{y})\label{P-prim}
\end{align}
with $\kappa$ given by (\ref{threshold}). Moreover, the reinsurance strategy generated by $\hat{\Xi}$ through (\ref{dynamic1}) is the insurer's equilibrium reinsurance strategy.
\end{theorem}
\begin{proof}
The proof boils down to verifying that $(V^{I},g)$ given by (\ref{VI}) and (\ref{g}) validate the proof of Theorem \ref{verif-insurer}. 

Unfortunately, $(V^{I},g)$ does not satisfy Condition (1) of Theorem \ref{verif-insurer}. However, we will show that the proof in Theorem \ref{verif-insurer} is still valid.

Specifically, the proofs of {\bf Step 1} and {\bf Step 2} in Theorem \ref{verif-insurer} only require that
\begin{align*}
		&V^{I}\in C^{2,1,1}(\overline{W^{\hat{\Xi}}})\bigcap C^{1,0,1}(P^{\hat{\Xi}}),\\
		&g\in C^{2,1,1}(\overline{W^{\hat{\Xi}}})\bigcap C^{1,0,1}(P^{\hat{\Xi}}).
	\end{align*}
 In {\bf Step 3}, Condition (1) validates the use of  It\^{o}-Tanaka-Meyer's formula in (\ref{vt6})$\sim$(\ref{vt8}) as the state-time-control triple moves from $(x-\Delta\eta_{t},t,y+\Delta\eta_{t})$ to $(X^{\eta}_{(t+h)-},t+h,\eta_{(t+h)-})$. By the form of $(W^{\hat{\Xi}}, P^{\hat{\Xi}})$, we deduce that the path is within $W^{\hat{\Xi}}$ or $P^{\hat{\Xi}}$ for $h$ sufficiently small. Hence the It\^{o}-Tanaka-Meyer's formula is still valid.

By (\ref{W-prim})$\sim$(\ref{P-prim}) and the linear structures of $V^{I}$ and $g$ given in (\ref{VI})$\sim$(\ref{g}), one can directly verify that $(V^{I},g)$ satisfies Conditions (3)$\sim$(4) of Theorem \ref{verif-insurer}. Then the proof follows using Theorem \ref{verif-insurer}.
 \end{proof}
With a little abuse of notation, we refer to the time regions when $y<\bar{y}$ when we mention the waiting and the purchasing regions in the rest of the paper. Moreover, we refer to the divisions of $[0, T]$ when we mention reinsurance laws. And we write $\hat{\Xi}=(W^{\hat{\Xi}},P^{\hat{\Xi}})$ where
\begin{align}
&W^{\hat{\Xi}}=\{t|c(t)>\kappa(t)\},\label{W-simple}\\
&P^{\hat{\Xi}}=\{t|c(t)\le \kappa(t)\}.\label{P-simple}
\end{align}

\section{The Reinsurer's Problem}
\label{sec-reinsurer}

This section aims to solve the reinsurer's problem. In Subsection \ref{subsec-eec}, we deduce the equivalent equilibrium conditions to those given in Definition \ref{eqps}. As a byproduct, we find that the reinsurer's problem can be transformed into a time-selection problem, whose equilibrium can be defined similarly to time-inconsistent stopping problems. Further, we find that the equilibrium conditions are equivalent to those of the time-selection problem. In Subsection \ref{subsec-edr}, we apply the equivalent equilibrium conditions to obtain the reinsurer's equilibrium strategy. The reinsurer's equilibrium strategy can be explicitly derived in different regions of additionally introduced parameters.

\subsection{Equivalent equilibrium conditions}
\label{subsec-eec}
In this subsection, we analyze the equilibrium conditions of Definition \ref{eqps} and find their equivalent counterparts. We also transfer the reinsurer's problem to a time-selection problem with equivalent equilibrium conditions.

According to Theorem \ref{theo-ins}, the insurer's equilibrium reinsurance strategy is to purchase the amount of $\bar{y}-y$ immediately the first time $c(t)\le \kappa(t)$, i.e., the purchasing time given by $p(t;c)$ in (\ref{ptc}). For the reinsurer to maximize his profit, the premium at $p(t;c)$ should exactly be set to $\kappa(p(t;c))$. Then the dependence of $J^{R}$ on $c$ is only through $p(\cdot;c)$. Note that there is no dependence of $J^{R}$ on $x$, then there exists function $K(t,y,z,p)$ such that
\begin{equation*}
J^{R}(x,t,y,z;c)=K\big(t,y,z,p(t;c)\big).
\end{equation*}

At this point, one can observe that the reinsurer's problem is transformed into a time-selection problem, which is time-inconsistent and similar to time-inconsistent stopping problems. We introduce the following definitions.

\begin{definition}[the reinsurer's time-selection problem and its equilibrium selected time]
We call the following problem the reinsurer's time-selection problem:
\begin{algorithm}[ht]
\floatname{algorithm}{Problem}
\caption{The reinsurer's time-selection problem}
Given any initial $(t,y,z)$, the reinsurer chooses a time $p(t;c)\in[t, T]\bigcup\{\infty\}$ at which his risk exposure increases by $\bar{y}-y$ and he obtains a premium of $(\bar{y}-y)\kappa\big(p(t;c)\big)$. His objective is to minimize $K\big(t,y,z,p(t;c)\big)$ over all choices of $p(t;c)$.
\end{algorithm}
\end{definition} \label{eqts}
We can define the equilibrium of the time-selection problem in a way similar to time-inconsistent stopping problems.
\begin{definition}[equilibrium selected time]
We call $p(t;c)$ an equilibrium selected time of the reinsurer's time-selection problem if the following two conditions hold.\\
(a) The selected time is no worse than now, i.e.,
\begin{equation*}
K(t,y,z,t)\ge K\big(t,y,z,p(t;c)\big),\forall (t,y,z)\in[0,T]\times[0,\bar{y}]\times\mathbb{R}.
\end{equation*}
(b) For any $h>0$, define the perturbed time $p^{h}(t;c)$ where $[t,t+h)$ cannot be selected, i.e.,
\begin{equation*}
p^{h}(t;c)=\max\{p(t;c),t+h\}.
\end{equation*}
Then it holds that
\begin{equation*}
\lim\limits_{h\rightarrow 0}\frac{K(t,y,z,p^{h}(t;c))-K\big(t,y,z,p(t;c)\big)}{h}\ge 0.
\end{equation*}
\end{definition}

We will later show that the equilibrium conditions in Definition \ref{eqps} are equivalent to the equilibrium conditions in Definition \ref{eqts}. Now let us focus back on the reinsurer's problem and transform the equilibrium conditions in Definition \ref{eqps} to the arguments on $K$. For notational simplicity, we introduce the following definition.
\begin{definition}
Suppose that $(W^{\hat{\Xi}}, P^{\hat{\Xi}})$ is a division of $[0,T]$ where $W^{\hat{\Xi}}$ and $P^{\hat{\Xi}}$ denote the waiting and purchasing regions respectively.\\
(a) We call $(t_{1},t_{2})\subset W^{\hat{\Xi}}$ a sub waiting region if for any $\epsilon\ge 0$, both $t_{1}-\epsilon$ and $t_{2}+\epsilon$ are in $P^{\hat{\Xi}}$.\\
(b) We call $(t_{1},T]\subset W^{\hat{\Xi}}$ the last sub waiting region if for any $\epsilon\ge 0$, $t_{1}-\epsilon$ is in $P^{\hat{\Xi}}$. Similarly, we can define the first waiting region $[0,t_{2})$. It should be noted that both the last and the first waiting regions may not exist.\\
(c) We call $[t_{1},t_{2}]\subset P^{\hat{\Xi}}$ a sub purchasing region if for any $\epsilon>0$, both $t_{1}-\epsilon$ and $t_{2}+\epsilon$ are in $W^{\hat{\Xi}}$. We call a sub purchasing region $[t_{1},t_{2}]$ a last sub purchasing region if $t_{2}=T$.
\end{definition}

Then, by Definition \ref{eqps}, we deduce the following equivalent conditions. The proof relies on discussing different cases of perturbations and is straightforward. We omit the proof here.

\begin{lemma}
\label{equil-lemma}
Let $\hat{c}$ be an admissible premium strategy and $\hat{\Xi}=(W^{\hat{\Xi}}, P^{\hat{\Xi}})$ be the insurer's equilibrium reinsurance law given $\hat{c}$, which is given by (\ref{W-simple})$\sim$(\ref{P-simple}). Then $\hat{c}$ is an equilibrium premium strategy if and only if the following conditions hold.\\
(a) For any $t$ in any sub waiting region $(t_{1},t_{2})$ of $W^{\hat{\Xi}}$, it holds that for any $\delta(t,h)\in[0,h]$ and $(y,z)\in[0,\bar{y})\times\mathbb{R}$,
\begin{equation*}
\liminf\limits_{h\rightarrow 0}\frac{K\big(t,y,z,t+\delta(t,h)\big)-K\big(t,y,z,p(t;\hat{c})\big)}{h}\ge 0.
\end{equation*}
(b) For any $t$ in a last sub waiting region $(t_{1},T]$ of $W^{\hat{\Xi}}$, it holds that for any $\delta(t,h)\in[0,h]$, $(y,z)\in[0,\bar{y})\times\mathbb{R}$ and $t\in(t_{1},T)$,
\begin{equation*}
\liminf\limits_{h\rightarrow 0}\frac{K\big(t,y,z,t+\delta(t,h)\big)-K\big(t,y,z,p(t;\hat{c})\big)}{h}\ge 0,
\end{equation*}
and
\begin{equation*}
K(T,y,z,T)-K(T,y,z,\infty)\ge 0.
\end{equation*}
(c) For any $t$ in any sub purchasing region $[t_{1},t_{2}]$ of $P^{\hat{\Xi}}$ with $t_{2}<T$, it holds that for any  $\delta(t,h)\in[0,h]$ and $(y,z)\in[0,\bar{y})\times\mathbb{R}$,
\begin{equation*}
\liminf\limits_{h\rightarrow 0}\frac{K\big(t,y,z,t+\delta(t,h)\big)-K(t,y,z,t)}{h}\ge 0,
\end{equation*}
and
\begin{equation*}
K(t_{2},y,z,t_{2})\le K\big(t_{2},y,z,q(t_{2};\hat{c})\big),
\end{equation*}
where
\begin{equation*}
q(t;c):=\left\{
\begin{array}{l}
\min\big\{\tau\in(t,T]\big|c(\tau)\le \kappa(\tau)\big\}, \ {\rm if\ such}\ \tau\ {\rm exists},\\
+\infty,\ {\rm else}.
\end{array}
\right.
\end{equation*}
(d) For any $t$ in a last sub purchasing region $[t_{1},T]$ of $P^{\hat{\Xi}}$, it holds that for any  $\delta(t,h)\in[0,h]$, $(y,z)\in[0,\bar{y})\times\mathbb{R}$ and $t\in[t_{1},T)$,
\begin{equation*}
\liminf\limits_{h\rightarrow 0}\frac{K\big(t,y,z,t+\delta(t,h)\big)-K(t,y,z,t)}{h}\ge 0.
\end{equation*}
And it holds that
\begin{equation*}
K(T,y,z,T)\le K(T,y,z,\infty).
\end{equation*}
\end{lemma}

We observe that Lemma \ref{equil-lemma} transfers all equilibrium conditions on $\hat{c}$ to conditions on $(W^{\hat{\Xi}}, P^{\hat{\Xi}})$ only, where we notice that $p(t;\hat{c})$ is given by (\ref{ptc}). Hence, we introduce the following useful definition.

\begin{definition}[reinsurer's equilibrium reinsurance law]
We call a division $(W^{\hat{\Xi}}, P^{\hat{\Xi}})$ of $[0, T]$ a reinsurer's equilibrium reinsurance law if it satisfies all conditions in Lemma \ref{equil-lemma}.
\end{definition}

With any reinsurer's equilibrium reinsurance law $(W^{\hat{\Xi}}, P^{\hat{\Xi}})$, we can generate equilibrium premium strategies as follows. 
\begin{equation}
\label{equil-premium}
\hat{c}(t)=\left\{
\begin{array}{l}
\kappa(t), \ t\in P^{\hat{\Xi}},\\
{\rm any\ value\ strictly\ larger\ than\ } \kappa(t), \ t\in W^{\hat{\Xi}}.
\end{array}
\right.
\end{equation}
It is implied that there are many equilibrium premium strategies. Specifically, the high price in the waiting region is not unique according to (\ref{equil-premium}).

To further apply Lemma \ref{equil-lemma}, we solve the function $K$ explicitly using It\^{o} calculus.
\begin{proposition} \label{P46}
The function $K$ is given by
\begin{align}
K(t,y,z,p)=&ze^{-b^{R}(T-t)}+\frac{a^{R}}{b^{R}}\Big[1-e^{-b^{R}(T-t)}\Big]+\gamma^{R}\frac{(\sigma^{R})^{2}}{2b^{R}}\Big[1-e^{-2b^{R}(T-t)}\Big]\notag\\
&+\bigg\{\Big[e^{-b^{R}(T-p)}-\frac{\theta^{R}}{\theta^{I}}e^{(-b^{I}+\rho^{R}-\rho^{I})(T-p)}\Big](\bar{y}-y)\bigg\}1_{p<\infty}.\label{K}
\end{align}
\end{proposition}
\begin{proof}
We consider two cases $p<\infty$ and $p=\infty$ separately.\\
{\bf Case 1: $p<\infty$\\}
Based on the definition of $K(t,y,z,p) $, we have 
\begin{equation*}
K(t,y,z,p)=\mathbb{E}_{z,t}Z^{\hat{\xi}}_{T}+\gamma^{R}\var_{z,t}Z^{\hat{\xi}}_{T}-\theta^{R}\mathbb{E}_{z,t}\int_{t}^{T}e^{\rho^{R}(T-r)}c(r)d\hat{\xi}_{r},
\end{equation*}
where $\hat{\xi}_{r}=y+(\bar{y}-y)1_{r\ge p},\forall r\in[t,T]$ and $Z^{\hat{\xi}}$ is the controlled risk exposure process under $\hat{\xi}$. Then
\begin{equation*}
\left\{
\begin{array}{l}
dZ^{\hat{\xi}}_{r}=(a^{R}-b^{R}Z^{\hat{\xi}}_{r})dr+\sigma^{R} dB^{R}_{r},\ r\in[t,p),\\
Z^{\hat{\xi}}_{t-}=z,
\end{array}
\right.
\end{equation*}
and 
\begin{equation*}
Z^{\hat{\xi}}_{p-}=ze^{-b^{R}(p-t)}+\frac{a^{R}}{b^{R}}[1-e^{-b^{R}(p-t)}]+\sigma^{R}\int_{t}^{p}e^{-b^{R}(p-r)}dB^{R}_{r}.
\end{equation*}
Similarly,
 we have
\begin{equation*}
\left\{
\begin{array}{l}
dZ^{\hat{\xi}}_{r}=(a^{R}-b^{R}Z^{\hat{\xi}}_{r})dr+\sigma^{R} dB^{R}_{r},\ r\in[p,T],\\
Z^{\hat{\xi}}_{p}=Z^{\hat{\xi}}_{p-}+\bar{y}-y,
\end{array}
\right.
\end{equation*}
 and 
\begin{equation*}
Z^{\hat{\xi}}_{T}=ze^{-b^{R}(T-t)}+(\bar{y}-y)e^{-b^{R}(T-p)}+\frac{a^{R}}{b^{R}}[1-e^{-b^{R}(T-t)}]+\sigma^{R}\int_{t}^{T}e^{-b^{R}(T-r)}dB^{R}_{r}.
\end{equation*}
Hence, for $p<\infty$, we have
\begin{align*}
K(t,y,z,p)=&ze^{-b^{R}(T-t)}+\frac{a^{R}}{b^{R}}\Big[1-e^{-b^{R}(T-t)}\Big]+\gamma^{R}\frac{(\sigma^{R})^{2}}{2b^{R}}\Big[1-e^{-2b^{R}(T-t)}\Big]\\
&+\Big[e^{-b^{R}(T-p)}-\frac{\theta^{R}}{\theta^{I}}e^{(-b^{I}+\rho^{R}-\rho^{I})(T-p)}\Big](\bar{y}-y).
\end{align*}
{\bf Case 2: $p=\infty$\\}
In this case, there is no reinsurance purchase. The terminal reinsurer's risk exposure process following the mean-reverting dynamic,
\begin{equation*}
Z_{T}=ze^{-b^{R}(T-t)}+\frac{a^{R}}{b^{R}}[1-e^{-b^{R}(T-t)}]+\sigma^{R}\int_{t}^{T}e^{-b^{R}(p-r)}dB^{R}_{r}.
\end{equation*}
Hence, for $p=\infty$
\begin{equation*}
K(t,y,z,p)=ze^{-b^{R}(T-t)}+\frac{a^{R}}{b^{R}}\Big[1-e^{-b^{R}(T-t)}\Big]+\gamma^{R}\frac{(\sigma^{R})^{2}}{2b^{R}}\Big[1-e^{-2b^{R}(T-t)}\Big].
\end{equation*}
Combining the two cases concludes the proof.
\end{proof}
Using Proposition \ref{P46} and Lemma \ref{equil-lemma}, we obtain the following proposition that further sketches the reinsurer's equilibrium reinsurance law. The proof is straightforward by substituting (\ref{K}) into  Conditions in Lemma \ref{equil-lemma} and we omit them  here.

\begin{proposition}
\label{furthersketch}
Define the continuously differentiable function
\begin{equation*}
\phi(t)=e^{-b^{R}(T-t)}-\frac{\theta^{R}}{\theta^{I}}e^{(-b^{I}+\rho^{R}-\rho^{I})(T-t)}.
\end{equation*}
Then any division $(W^{\hat{\Xi}}, P^{\hat{\Xi}})$ of $[0, T]$ is a reinsurer's equilibrium reinsurance law if and only if the following conditions hold (In the following, the sub waiting and purchasing regions are w.r.t. $W^{\hat{\Xi}}$ and $P^{\hat{\Xi}}$).\\
(a) If $(t_{1},t_{2})$ is a sub waiting region, then
\begin{equation*}
\phi(t)\ge \phi(t_{2}),\forall t\in(t_{1},t_{2}).
\end{equation*}
(b) If $(t_{1},T]$ is the last sub waiting region , then
\begin{equation*}
\phi(t)\ge 0,\forall t\in(t_{1},T].
\end{equation*}
(c)  If $[t_{1},t_{2}]$ is a sub purchasing region with a neighboring sub waiting region $(t_{2},t_{3})$ where $t_{3}<T$, then
\begin{equation*}
\phi'(t)\ge 0,\forall t\in[t_{1},t_{2}],
\end{equation*}
and 
\begin{equation*}
\phi(t_{2})\le \phi(t_{3}).
\end{equation*}
(d) If $[t_{1},t_{2}]$ is a last sub purchasing region with $t_{2}=T$ or a sub purchasing region  with a neighboring last sub waiting region $(t_{2}, T]$, then
\begin{equation*}
\phi'(t)\ge 0,\forall t\in[t_{1},T],
\end{equation*}
and 
\begin{equation*}
\phi(t_{2})\le 0.
\end{equation*}
\end{proposition}

As a byproduct, we find that substituting (\ref{K}) into the conditions in Definition \ref{eqts} leads to the same conditions in the last proposition. Hence, the equilibrium conditions in Definition \ref{eqps} are equivalent to the equilibrium conditions in Definition \ref{eqts}, i.e., we have the following corollary.
\begin{corollary}
    $\hat{c}$ is an equilibrium premium strategy of the reinsurer's problem if and only if $p(t;\hat{c})$ given by (\ref{ptc}) is an equilibrium selected time of the reinsurer's time-selection problem.
\end{corollary}

\subsection{Equilibrium in different regions}
\label{subsec-edr}
In this subsection, we use Proposition \ref{furthersketch} to find the explicit form of the reinsurer's equilibrium reinsurance law, which can further generate equilibrium premium strategies through (\ref{equil-premium}).

First of all, let us introduce two important parameters for further analysis.

\begin{definition}
\label{para}
We refer $\theta^{I}$ and $\theta^{R}$ to the reinsurance preferences of the insurer and the reinsurer. And we refer $b^{I}+\rho^{I}$ and $b^{R}+\rho^{R}$ to the reversion-discount sums of the insurer and the reinsurer. We define the premium preference ratio $r$ and the reversion-discount difference $d$ by
\begin{align*}
&r:=\frac{\theta^{I}}{\theta^{R}},\\
&d:=(b^{I}+\rho^{I})-(b^{R}+\rho^{R}).
\end{align*}
\end{definition}
Then, using Proposition \ref{furthersketch}, we can obtain the reinsurer's equilibrium reinsurance law in different cases of the reversion-discount difference and the premium preference ratio.

(1) The first class of cases is when $d=0$. In this class of cases, $\phi$ keeps constant. We further have the following three cases.

(1a) If $r=1$, then $\phi$ is always zero and any $(W^{\hat{\Xi}},P^{\hat{\Xi}})$ as a division of $[0,T]$ is a reinsurer's equilibrium reinsurance law.

(1b) If $r<1$, then $\phi$ is strictly decreasing with $\phi(T)<0$. We deduce that
\begin{equation*}
\left\{
\begin{array}{l}
W^{\hat{\Xi}}=[0,T),\\
P^{\hat{\Xi}}=\{T\},
\end{array}
\right.
\end{equation*}
is the reinsurer's equilibrium reinsurance law.

(1c) If $r>1$, then $\phi$ is strictly increasing with $\phi(0)>0$ and 
\begin{equation*}
\left\{
\begin{array}{l}
W^{\hat{\Xi}}=[0,T],\\
P^{\hat{\Xi}}=\emptyset,
\end{array}
\right.
\end{equation*}
is the reinsurer's equilibrium reinsurance law.

(2) The second class of cases is when $d\neq 0$. We further have the following four cases according to the monotonicity of $\phi$.

(2a) $\phi$ is increasing in $[0,T]$, which holds if
\begin{equation*}
r\ge (1+\frac{d}{b^{R}})\max\{1,e^{-Td}\}.
\end{equation*}
Then the reinsurer's equilibrium reinsurance law $(W^{\hat{\Xi}},P^{\hat{\Xi}})$ is given by
\begin{equation*}
\left\{
\begin{array}{l}
\left\{
\begin{array}{l}
W^{\hat{\Xi}}=[0,T],\\
P^{\hat{\Xi}}=\emptyset.
\end{array}
\right.
,\ r> e^{-Td},\vspace{1ex}\\
\left\{
\begin{array}{l}
W^{\hat{\Xi}}=\emptyset,\\
P^{\hat{\Xi}}=[0,T].
\end{array}
\right.
,\ r\le 1,\vspace{1ex}\\
\left\{
\begin{array}{l}
W^{\hat{\Xi}}=\big(T+\frac{1}{d}\ln(r),T\big],\\
P^{\hat{\Xi}}=\big[0,T+\frac{1}{d}\ln(r)\big].
\end{array}
\right.
,\ {\rm else}.
\end{array}
\right.
\end{equation*}

(2b) $\phi$ is decreasing in $[0,T]$, which holds if
\begin{equation*}
r\le (1+\frac{d}{b^{R}})\min\{1,e^{-Td}\}.
\end{equation*}
Then the reinsurer's equilibrium reinsurance law $(W^{\hat{\Xi}},P^{\hat{\Xi}})$ is given by
\begin{equation*}
\left\{
\begin{array}{l}
\left\{
\begin{array}{l}
W^{\hat{\Xi}}=[0,T],\\
P^{\hat{\Xi}}=\emptyset,
\end{array}
\right.
 ,\ r\ge 1,\vspace{1ex}\\
 \left\{
\begin{array}{l}
W^{\hat{\Xi}}=[0,T),\\
P^{\hat{\Xi}}=\{T\},
\end{array}
\right.
,\ r\le 1,
\end{array}
\right.
\end{equation*}
where both laws are reinsurer's equilibrium reinsurance laws when $r=1$.

(2c) $\phi$ is first increasing then decreasing in $[0,T]$, which holds if
\begin{equation*}
(1+\frac{d}{b^{R}})e^{-Td}<r< 1+\frac{d}{b^{R}}.
\end{equation*}
In this case, we have $b^{I}-\rho^{R}+\rho^{I}>b^{R}>0$ and then $\lim\limits_{t\rightarrow-\infty}\phi(t)=0$, which indicates that $\phi(0)>0$. Then the reinsurer's equilibrium reinsurance law $(W^{\hat{\Xi}},P^{\hat{\Xi}})$ is given by
\begin{equation*}
\left\{
\begin{array}{l}
\left\{
\begin{array}{l}
W^{\hat{\Xi}}=[0,T],\\
P^{\hat{\Xi}}=\emptyset,
\end{array}
\right.
 ,\ r\ge 1,\vspace{1ex}\\
 \left\{
\begin{array}{l}
W^{\hat{\Xi}}=[0,T),\\
P^{\hat{\Xi}}=\{T\},
\end{array}
\right.
,\ r\le 1,
\end{array}
\right.
\end{equation*}
where both laws are reinsurer's equilibrium reinsurance laws when $r=1$.

(2d) $\phi$ is first decreasing then increasing in $[0,T]$, which holds if
\begin{equation*}
1+\frac{d}{b^{R}}<r< (1+\frac{d}{b^{R}})e^{-Td}.
\end{equation*}
In this case, we have $b^{R}>b^{I}-\rho^{R}+\rho^{I}>0$ and then $\lim\limits_{t\rightarrow-\infty}\phi(t)=0$, which indicats that $\phi(0)<0$.
Then the reinsurer's equilibrium reinsurance law $(W^{\hat{\Xi}},P^{\hat{\Xi}})$ is given by
\begin{equation*}
\left\{
\begin{array}{l}
\left\{
\begin{aligned}
W^{\hat{\Xi}}=&\big([0,T+\frac{1}{d}\ln(\frac{b^{R}r}{b^{R}+d})\big),\\
P^{\hat{\Xi}}=&\big[T+\frac{1}{d}\ln(\frac{b^{R}r}{b^{R}+d}),T\big].
\end{aligned}
\right.
 ,\ r\le 1,\vspace{1ex}\\
 \left\{
\begin{aligned}
W^{\hat{\Xi}}=&\big[0,T+\frac{1}{d}\ln(\frac{b^{R}r}{b^{R}+d})\big)\bigcup\big(T+\frac{1}{d}\ln(r),T\big],\\
P^{\hat{\Xi}}=&\big[T+\frac{1}{d}\ln(\frac{b^{R}r}{b^{R}+d}),T+\frac{1}{d}\ln(r)\big].
\end{aligned}
\right.
,\ r>1.
\end{array}
\right.
\end{equation*}

The above analysis implies several important boundaries $d=0$, $r=1$, $r=e^{-Td}$, $r=1+\frac{d}{b^{R}}$ and $r=(1+\frac{d}{b^{R}})e^{-Td}$. For a clearer presentation of the above analysis, we label the regions separated by these boundaries.

\begin{definition}
Define the following eight sub regions of $\mathbb{R}\times\mathbb{R}^{+}$:
\begin{align*}
&\uppercase\expandafter{\romannumeral 1}=\Big\{(d,r)\Big|r>1, r>e^{-Td}\Big\},\\
&\uppercase\expandafter{\romannumeral 2}=\Big\{(d,r)\Big|r>1, r\le e^{-Td}\Big\},\\
&\uppercase\expandafter{\romannumeral 3}=\Big\{(d,r)\Big|r\le 1, r\ge \big(1+\frac{d}{b^{R}}\big)e^{-Td}, d<0\Big\},\\
&\uppercase\expandafter{\romannumeral 4}=\Big\{(d,r)\Big|r\le 1, r<\big(1+\frac{d}{b^{R}}\big)e^{-Td} , r\ge 1+\frac{d}{b^{R}}\Big\},\\
&\uppercase\expandafter{\romannumeral 5}=\Big\{(d,r)\Big|r<1+\frac{d}{b^{R}},r<1\},\\
&\uppercase\expandafter{\romannumeral 6}=\Big\{(d,r)\Big|r=1,d>0\Big\},\\
&\uppercase\expandafter{\romannumeral 7}=\Big\{(0,1)\Big\},\\
&\uppercase\expandafter{\romannumeral 8}=\Big\{(d,r)\Big|1<r<(1+\frac{d}{b^{R}})e^{-Td},d<0\Big\}.
\end{align*}
Among the eight regions, we call regions $\uppercase\expandafter{\romannumeral 1}\sim\uppercase\expandafter{\romannumeral 5}$ and $\uppercase\expandafter{\romannumeral 8}$ area regions; the region $\uppercase\expandafter{\romannumeral 6}$ consists of 
 a line and we call it a line region; the region $\uppercase\expandafter{\romannumeral 7}$ consists of a single point and we call it a point region; the area region $\uppercase\expandafter{\romannumeral 8}$ only appears if the time horizon is long enough, i.e.,  $T>\frac{1}{b^{R}}$.
\end{definition}
These regions are shown in Figure \ref{regions1} for the short-term case $T\le\frac{1}{b^{R}}$ and Figure \ref{regions2} for the long-term case $T>\frac{1}{b^{R}}$.

\begin{figure}[ht]
\centering
\includegraphics[width=4in, keepaspectratio]{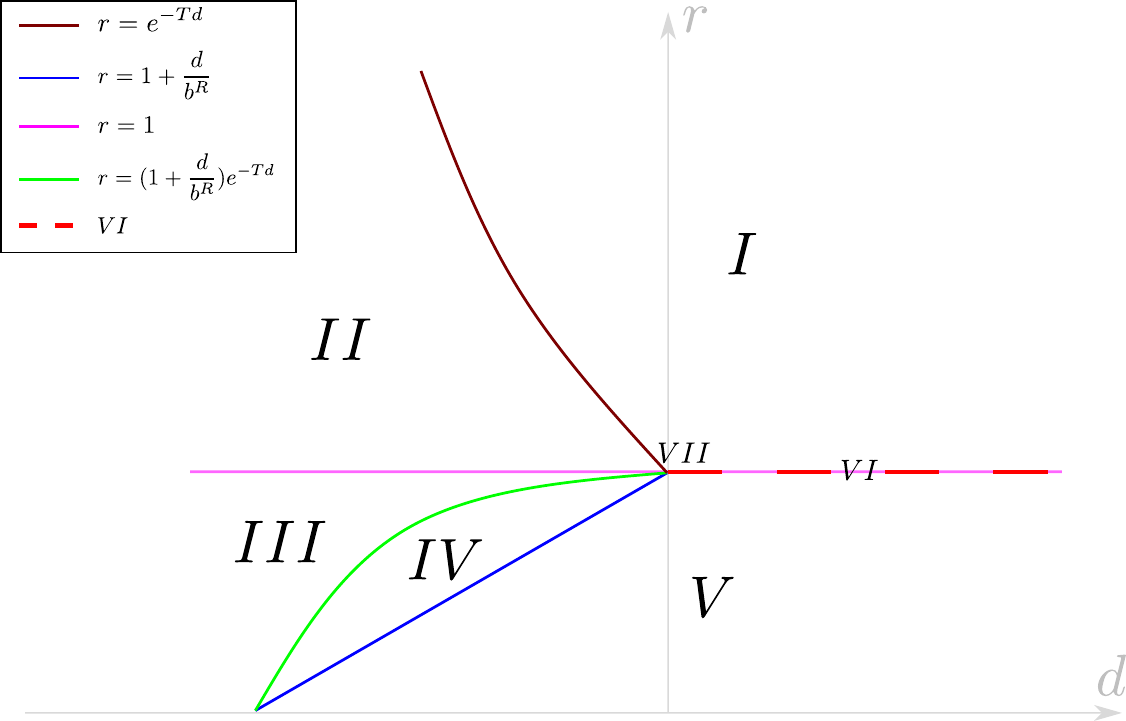}\\
\caption{An illustration of the seven regions for the short-term case $T\le \frac{1}{b^{R}}$ (region $\uppercase\expandafter{\romannumeral 8}$ does not exist in the short-term case). The boundary lines are shown in solid lines. The line region $\uppercase\expandafter{\romannumeral 6}$ is marked as a dashed line.}
\label{regions1}
\end{figure}

\begin{figure}[ht]
\centering
\includegraphics[width=4in, keepaspectratio]{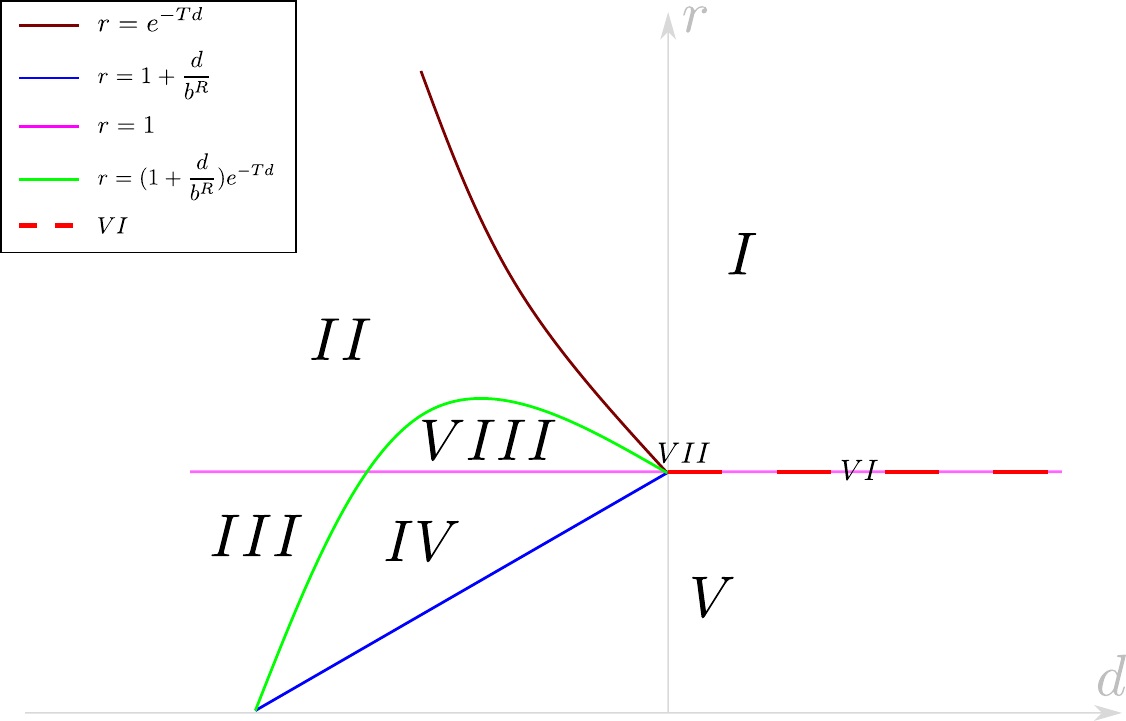}\\
\caption{An illustration of the eight regions for the long-term case $T> \frac{1}{b^{R}}$. The boundary lines are shown in solid lines. The line region $\uppercase\expandafter{\romannumeral 6}$ is marked as a dashed line.}
\label{regions2}
\end{figure}
With the above definitions, we can summarize the main result of this section in the following theorem. The proof directly follows from the analysis in this subsection and is omitted here.
\begin{theorem}
\label{theo-rei}
The reinsurer's equilibrium reinsurance law $(W^{\hat{\Xi}},P^{\hat{\Xi}})$ is given in Table \ref{rerl}. Moreover, given any reinsurer's equilibrium reinsurance law $(W^{\hat{\Xi}},P^{\hat{\Xi}})$, any premium strategy generated by $(W^{\hat{\Xi}},P^{\hat{\Xi}})$ according to (\ref{equil-premium}) is an equilibrium premium strategy.
\end{theorem}

\begin{table}[ht]
\centering
\caption{Reinsurer's equilibrium reinsurance laws in nine regions.}
\label{rerl}
\begin{tabular}{cl}
\hline
\textbf{Regions}           & 
\multicolumn{1}{c}{\textbf{Reinsurer's equilibrium reinsurance laws}} 
\\ 
\hline
\specialrule{0em}{3pt}{3pt}
$\uppercase\expandafter{\romannumeral 1}$  & 
\qquad\quad
$\left\{\begin{array}{l}
W^{\hat{\Xi}}=\big[0, T\big]\\
P^{\hat{\Xi}}=\emptyset
\end{array}\right.$
\\ 
\specialrule{0em}{1.5pt}{1.5pt}
%\hline
\specialrule{0em}{1.5pt}{1.5pt}
$\uppercase\expandafter{\romannumeral 2}$  & 
\qquad\quad
$\left\{\begin{array}{l}
W^{\hat{\Xi}}=\big(T+\frac{1}{d}\ln(r),T\big]\vspace{1.5pt}\\
P^{\hat{\Xi}}=\big[0,T+\frac{1}{d}\ln(r)\big]
\end{array}\right.$
\\ 
\specialrule{0em}{1.5pt}{1.5pt}
%\hline
\specialrule{0em}{1.5pt}{1.5pt}
$\uppercase\expandafter{\romannumeral 3}$  &    
\qquad\quad
$\left\{\begin{array}{l}
W^{\hat{\Xi}}=\emptyset\\
P^{\hat{\Xi}}=\big[0,T\big]
\end{array}\right.$
\\ 
\specialrule{0em}{1.5pt}{1.5pt}
%\hline
\specialrule{0em}{1.5pt}{1.5pt}
$\uppercase\expandafter{\romannumeral 4}$  &   
\qquad\quad
$\left\{\begin{array}{l}
W^{\hat{\Xi}}=\big[0,T+\frac{1}{d}\ln(\frac{b^{R}r}{b^{R}+d})\big)\vspace{1.5pt}\\
P^{\hat{\Xi}}=\big[T+\frac{1}{d}\ln(\frac{b^{R}r}{b^{R}+d}),T\big]
\end{array}\right.$
\\ 
\specialrule{0em}{1.5pt}{1.5pt}
%\hline
\specialrule{0em}{1.5pt}{1.5pt}
$\uppercase\expandafter{\romannumeral 5}$  &   
\qquad\quad
$\left\{\begin{array}{l}
W^{\hat{\Xi}}=[0,T)\\
P^{\hat{\Xi}}=\{T\}
\end{array}\right.$
\\ 
\specialrule{0em}{1.5pt}{1.5pt}
%\hline
\specialrule{0em}{1.5pt}{1.5pt}
$\uppercase\expandafter{\romannumeral 6}$  &   
\qquad\quad
$\left\{\begin{array}{l}
W^{\hat{\Xi}}=[0,T)\\
P^{\hat{\Xi}}=\{T\}
\end{array}\right.$
 or 
 $\left\{\begin{array}{l}
W^{\hat{\Xi}}=[0,T]\\
P^{\hat{\Xi}}=\emptyset
\end{array}\right.$
\\ 
\specialrule{0em}{1.5pt}{1.5pt}
%\hline
\specialrule{0em}{1.5pt}{1.5pt}
$\uppercase\expandafter{\romannumeral 7}$  &     
\qquad\quad
Arbitrary division of $[0,T]$
\\ 
\specialrule{0em}{1.5pt}{1.5pt}
%\hline
\specialrule{0em}{1.5pt}{1.5pt}
$\uppercase\expandafter{\romannumeral 8}$ &   
\qquad\quad
$\left\{\begin{array}{l}
W^{\hat{\Xi}}=\big[0,T+\frac{1}{d}\ln(\frac{b^{R}r}{b^{R}+d}\big)\bigcup\big(T+\frac{1}{d}\ln(r),T\big]\vspace{1.5pt}\\
P^{\hat{\Xi}}=\big[T+\frac{1}{d}\ln(\frac{b^{R}r}{b^{R}+d},T+\frac{1}{d}\ln(r)\big]
\end{array}\right.$
\\ 
\specialrule{0em}{3pt}{3pt}
\hline
\end{tabular}
\end{table}

\section{An interpretation of the equilibrium strategies}
\label{sec-interpretation}
In this section, we aim to give an interpretation of the equilibrium reinsurance strategy given in Theorem \ref{theo-ins} and the equilibrium premium strategy given in Theorem \ref{theo-rei}. Also, we investigate the impact of various parameters on the equilibrium strategies.

\subsection{Equilibrium reinsurance strategy}

The insurer's equilibrium reinsurance strategy is relatively simple. 

According to Theorem \ref{theo-ins}, the insurer's equilibrium reinsurance law when $y<\bar{y}$, i.e., when there is a choice of reinsurance, is independent of the risk exposure and accumulated reinsurance coverage. The insurer's equilibrium reinsurance law indicates that reinsurance should be purchased if and only if the reinsurance premium is lower than the threshold $\kappa(t)$ given by (\ref{threshold}) at that time.

As a result, the insurer's equilibrium reinsurance strategy, which is generated by the insurer's equilibrium reinsurance law by solving a Skorohod reflection problem (\ref{dynamic1}), is to purchase all remaining amount of reinsurance the first time that the reinsurance premium falls below the time-dependent threshold, which is given by (\ref{ptc}).

The equilibrium reinsurance strategy supports the view that reinsurance contracts should be signed at discrete times rather than continuously adjusted. Specifically, in our problem, the contract is signed just once and it fulfills all reinsurance demands the insurer needs during the whole time horizon. For the insurer, all she needs to do is to choose the best time to sign a once-for-all reinsurance contract. The best time is simply the first time with a relatively low reinsurance premium rate. She does not need to know the reinsurer's global premium strategy in $[0, T]$ to make the decision. In a more realistic scene, neither the reinsurer would disclose his global premium strategy nor would he follow his pre-set global premium strategy. As a result, the insurer only has access to the reinsurance premium rate up to the current time. From this point of view, the insurer's equilibrium reinsurance strategy is practically reasonable and easy to apply.

The purchasing time of the insurer, i.e., the time to sign the reinsurance contract, depends on both the reinsurer's premium strategy and the time-dependent threshold $\kappa(t)$. Recall that 
\begin{equation*}
\kappa(t)=\frac{1}{\theta^{I}}e^{-(b^{I}+\rho^{I})(T-t)},\forall t\in[0,T],
\end{equation*}
we arrive at the following conclusions of sensitivity analysis. First, the larger $\theta^{I}$, the more weight put on the reinsurance premium cost, and the later the purchasing time because the insurer waits for a better opportunity with a relatively lower premium cost. Second, the larger $b^{I}$, i.e., the faster the reverting speed of risk exposure, the later the purchasing time. The cause is likely that the growth of the risk process can be adjusted to the mean value more quickly, which reduces the need to adjust the risk through reinsurance. Third, the larger the discount rate $\rho^{I}$, i.e., the more profit it has to keep the cash in hand, the later the purchasing time. This is similar to the case with larger $\theta^{I}$, where in both cases, purchasing reinsurance is more costly for the insurer and results in a more cautious reinsurance decision. However, it is observed that the parameters $a^{I}$, $\sigma^{I}$ and $\gamma^{I}$, i.e., the mean growth rate of risk, the volatility of risk and the weight put on variance, do not influence the equilibrium reinsurance strategy. It is likely that reinsurance does not affect the variance of terminal risk exposure.

\subsection{Equilibrium premium strategy}

The reinsurer's equilibrium premium strategy is rather complicated because he needs to consider the impact of his premium strategy on the insurer's decision-making. As a result, the parameters related to both the insurer and the reinsurer can influence the equilibrium premium strategy.

As the equilibrium reinsurance strategy is generated by the insurer's equilibrium reinsurance law, the equilibrium premium strategy is generated by the reinsurer's equilibrium reinsurance law. However, according to (\ref{equil-premium}), the generation of the equilibrium premium strategy is not unique to a large extent. It is implied that the premium rate should be set as the highest acceptable price for the insurer if he wants the insurer to purchase reinsurance, i.e., the purchasing region; while the premium rate can be set arbitrarily as long as it exceeds the threshold for the insurer to purchase reinsurance if the reinsurer does not want the reinsurance to be purchased, i.e., the waiting region. The non-uniqueness in the waiting region of the reinsurer's equilibrium reinsurance law is natural from this perspective.

The separation of the waiting region and the purchasing region, or equivalently, when the reinsurer wants the insurer to purchase reinsurance, is rather complicated. There are eight cases, which are depicted in Figure \ref{regions1}$\sim$\ref{regions2} and Table \ref{rerl}.

Among the eight regions, the region $\uppercase\expandafter{\romannumeral 1}$ is the region where the reinsurer's reinsurance preference is relatively small enough in comparison with that of the insurer, i.e.,  $r$ is sufficiently large. In this case, the reinsurer does not intend to sell reinsurance to the insurer. It implies that there is no deal if the insurer considers the reinsurance premium to be relatively too expensive than the reinsurer does. 

At the other extreme, the reinsurer's reinsurance preference is relatively large enough, i.e., $r$ is sufficiently small. In this case, the reinsurer is certain that the insurer will accept the reinsurance contract and his main goal is to maximize the reinsurance premium he earns discounted to the terminal time. There are two sub-cases: one sub-case $d>-b^{R}$ corresponds to the region $\uppercase\expandafter{\romannumeral 5}$ where the high terminal price is more attractive and he intends to sell at the terminal time; the other sub-case $d\le -b^{R}$ belongs to the region $\uppercase\expandafter{\romannumeral 3}$, where the high discount rate is more attractive and he intends to sell immediately.

The region $\uppercase\expandafter{\romannumeral 1}$ and the region $\uppercase\expandafter{\romannumeral 5}$ shares a boundary, which is the line region $\uppercase\expandafter{\romannumeral 6}$. On the boundary, both the strategy resulting in no deal and the strategy resulting in a terminal purchase are preferred as equilibrium strategies.

We have explained part of the region $\uppercase\expandafter{\romannumeral 3}$ from the perspective of reinsurance preference. A better way to explain the region $\uppercase\expandafter{\romannumeral 3}$ is from the perspective of discount. We cannot explain it from the perspective of reversion because the boundary line $r=(1+\frac{d}{b^{R}})e^{-Td}$ moves as $b^{R}$ changes. To be specific, the combination of regions $\uppercase\expandafter{\romannumeral 3}$ and $\uppercase\expandafter{\romannumeral 2}$ corresponds to the case where the reinsurer's discount is sufficiently large, i.e., $\rho^{R}$ is large enough. In both cases, the reinsurer intends to sell reinsurance immediately because the high discount rate stimulates him to earn the reinsurance premium immediately at the cost of bearing a high level of risk exposure. However, in the case with the region $\uppercase\expandafter{\romannumeral 2}$ where the insurer's reinsurance preference dominates, the reinsurer stops the intention of selling reinsurance at the cut-off time $T+\frac{1}{d}\ln(r)$. This is because the insurer considers the reinsurance premium to be relatively more expensive than the reinsurer, but this case is not as extreme as the case with the region $\uppercase\expandafter{\romannumeral 1}$ where the reinsurer stops the intention immediately.

The region $\uppercase\expandafter{\romannumeral 4}$ corresponds to an intermediate case between those of the region $\uppercase\expandafter{\romannumeral 3}$ and the region $\uppercase\expandafter{\romannumeral 5}$. In the region $\uppercase\expandafter{\romannumeral 4}$, the reinsurer's discount is less dominant than that of the region $\uppercase\expandafter{\romannumeral 3}$. As a result, the reinsurer does not intend to sell reinsurance immediately but waits until the start-up time $T+\frac{1}{d}\ln(\frac{b^{R}r}{b^{R}+d})$. However, the reinsurer's reinsurance preference still dominates as in the region $\uppercase\expandafter{\romannumeral 3}$ and there is no cut-off time. If the reinsurer's discount is too small, then the start-up time is postponed to the terminal time and it turns out to be the case in the region $\uppercase\expandafter{\romannumeral 5}$.

From the explanation of the region $\uppercase\expandafter{\romannumeral 4}$, we know that there might be a start-up time if the reinsurer's discount is not dominant enough. It is also the case when the reinsurer's discount becomes less dominant in the region $\uppercase\expandafter{\romannumeral 2}$. However, here the start-up time $T+\frac{1}{d}\ln(\frac{b^{R}r}{b^{R}+d})$ can be negative and it only shows up above zero if $T>\frac{1}{b^{R}}$. When $T>\frac{1}{b^{R}}$, this case with both a start-up time and a cut-off time corresponds to the region $\uppercase\expandafter{\romannumeral 8}$. 

Before we interpret the last point region $\uppercase\expandafter{\romannumeral 7}$. We briefly discuss the limiting case where the reinsurer's reversion tends to infinity. This case cannot be observed directly from Figures \ref{regions1}$\sim$\ref{regions2} because both $d$ and $b^{R}$ tend to infinity. Using the fact that when fixing $d<0$,
\begin{align*}
&\lim\limits_{\Delta b^{R}\rightarrow+\infty}1+\frac{d-\Delta b^{R}}{b^{R}+\Delta b^{R}}=0,\\
&\lim\limits_{\Delta b^{R}\rightarrow+\infty}
\left (1+\frac{d-\Delta b^{R}}{b^{R}+\Delta b^{R}}\right)e^{-T(d-\Delta b^{R})}=\left\{
\begin{array}{l}
+\infty,\ {\rm if}\ d>-b^{R},\\
-\infty,\ {\rm if}\ d<-b^{R},
\end{array}
\right.
\end{align*}
we deduce that if the reinsurer's discount is large enough, i.e., $\rho^{R}\ge b^{I}+\rho^{I}$, then as the reinsurer's reversion tends to infinity, the limiting case would belong to regions $\uppercase\expandafter{\romannumeral 2}$ and $\uppercase\expandafter{\romannumeral 3}$; if the reinsurer's discount is not large enough, i.e., $\rho^{R}<b^{I}+\rho^{I}$, then as the reinsurer's reversion tends to infinity, the limiting case would belong to regions $\uppercase\expandafter{\romannumeral 4}$ and $\uppercase\expandafter{\romannumeral 8}$.

The last point region $\uppercase\expandafter{\romannumeral 7}$ stands for the indifferent case where the insurer's reinsurance preference and reversion-discount sum are both the same as those of the reinsurer. In this case, all strategies are indifferent to the reinsurer.

From the above interpretations, the reinsurer has different intentions of selling reinsurance in different cases, which reflects his equilibrium premium strategy. We summarize the reinsurer's intention of selling reinsurance in all six area regions in Table \ref{arearegion}. We only show the intentions in area regions because the line region $\uppercase\expandafter{\romannumeral 6}$ and the point region $\uppercase\expandafter{\romannumeral 7}$ can be interpreted as the boundary cases of their neighboring area regions. 

\begin{table}[ht]
\centering
\caption{The reinsurer's intention to sell reinsurance in the six cases of the area regions, together with whether there is a cut-off time of the intention before the terminal time, and the start-up time of the intention.}
\label{arearegion}
\begin{tabular}{cccc}
\hline
\textbf{\begin{tabular}[c]{@{}c@{}}area\\ region\end{tabular}} & \textbf{\begin{tabular}[c]{@{}c@{}}an intention \\ to sell reinsurance  \end{tabular}} & \textbf{\begin{tabular}[c]{@{}c@{}}a cut-off time of the intention\\  before the terminal time  \end{tabular}} & \textbf{\begin{tabular}[c]{@{}c@{}}start-up \\ time\end{tabular}} \\ \hline
$\uppercase\expandafter{\romannumeral 1}$                      & no                                                                                    & \textbackslash{}                                                                                              & \textbackslash{}                                                  \\
$\uppercase\expandafter{\romannumeral 2}$                      & yes                                                                                   & no                                                                                                            & before $T$                                              \\
$\uppercase\expandafter{\romannumeral 3}$                      & yes                                                                                   & no                                                                                                            & no                                                                \\
$\uppercase\expandafter{\romannumeral 4}$                      & yes                                                                                   & yes                                                                                                           & no                                                                \\
$\uppercase\expandafter{\romannumeral 5}$                      & yes                                                                                   & no                                                                                                            & $T$                                                    \\
$\uppercase\expandafter{\romannumeral 8}$                      & yes                                                                                   & yes                                                                                                           & before $T$                                              \\ \hline
\end{tabular}
\end{table}

According to Table \ref{arearegion}, the reinsurer is reluctant to sell reinsurance in the region $\uppercase\expandafter{\romannumeral 1}$. Among the regions where the reinsurer has an intention to sell reinsurance, the intention is strong in regions $\uppercase\expandafter{\romannumeral 2}$, $\uppercase\expandafter{\romannumeral 3}$ and $\uppercase\expandafter{\romannumeral 5}$ because no cut-off time before the terminal time guarantees a reinsurance deal. The three cases differ in their start-up times of intention, reflecting different strategies. The intention is relatively low in regions $\uppercase\expandafter{\romannumeral 4}$ and $\uppercase\expandafter{\romannumeral 8}$ and the start-up times also differ in these two cases. From the analysis of limiting cases, we find that the reinsurer's relatively high reinsurance preference and discount are signs of a strong intention to sell reinsurance. Moreover, the reinsurer's relatively low reinsurance preference is a sign of reluctance to sell reinsurance. The reinsurer's relatively low reversion or discount, however, is not a sign of reluctance because its effect can be overwhelmed by the effect of a relatively high reinsurance preference of the reinsurer. Moreover, the reinsurer's relatively high reversion is not a sign of strong intention because it only takes effect when the reinsurer has a relatively high discount. It implies that in the reinsurer's decision-making, the reinsurance preference is more dominant than the discount, while the discount is more dominant than reversion.

\section{Conclusion}
\label{sec-conclusion}

In this paper, we study a reinsurance Stackelberg game in a unified singular control framework considering reinsurance contracts in both continuous and discrete times. The reinsurer is the leader who develops a reinsurance premium strategy and the insurer is the follower seeking a reinsurance strategy. The objectives of the insurer and reinsurer are of MV type and the reinsurance contracts are irreversible. We obtain both equilibrium reinsurance and premium strategies in explicit forms. We find that a single once-for-all reinsurance contract is preferred under MV criterion and irreversibility and the game between the insurer and the reinsurer is on the signing time of the contract. Moreover, we give interpretations of the equilibrium strategies and the influence of parameters are analyzed. It is interesting to consider other or more complicated game structures in various scenes coupled with singular control and various causes of time-inconsistency. Limited to the scope of the paper, we leave these topics for future research.

\section*{Acknowledgments}
The authors acknowledge the support from the National Natural Science Foundation of China (Nos. 12271290, 11871036). The authors also thank the members of the group of Actuarial Science and Mathematical Finance at the Department of Mathematical Sciences, Tsinghua University for their feedbacks and useful conversations.

\bibliographystyle{plainnat}
\bibliography{references}

\end{document}